\documentclass[%
 reprint,
 amsmath,amssymb,
 aps,prl
]{revtex4-2}
\usepackage{graphicx}
\usepackage{dcolumn}
\usepackage{bm}

\usepackage{amsthm,amsfonts,amsmath,amssymb}

\usepackage{mathabx}

\newtheorem{theorem}{Theorem}

\newtheorem{lemma}{Lemma}
\newtheorem{remark}{Remark}

\newcommand{\Lopr}{\hat{\mathcal L}_{r}}
\newcommand{\Lopp}{\hat{\mathcal L}_{p}}

\makeatletter
\newcounter{savesection}
\newcounter{apdxsection}
\renewcommand\appendix{\par
  \setcounter{savesection}{\value{section}}%
  \setcounter{section}{\value{apdxsection}}%
  \setcounter{subsection}{0}%
  \gdef\thesection{\@Alph\c@section}}
\newcommand\unappendix{\par
  \setcounter{apdxsection}{\value{section}}%
  \setcounter{section}{\value{savesection}}%
  \setcounter{subsection}{0}%
  \gdef\thesection{\@arabic\c@section}}
\makeatother

\begin{document}

\preprint{APS/123-QED}

\title{Symplectic and Thermodynamically Consistent Molecular Dynamics\\ in the Frequency Domain}

\author{Kyunghoon Han}
\email{kyunghoon.h@gmail.com}
\author{Alexandre Tkatchenko}%
\author{Joshua T. Berryman}
\affiliation{%
 Department of Physics and Materials Science,
 University of Luxembourg, L-1511 Luxembourg City, Luxembourg}%

\date{\today}

\begin{abstract} 
We introduce Fourier integrator molecular dynamics (FIMD), a method for propagating selected vibrational motion of Hamiltonian systems stably and reversibly in time while analyzing and controlling dynamics in the frequency domain. This makes band selection and vibrational analysis features of the integrator rather than post-processing steps. We demonstrate the method with classical force fields, a machine-learned force field trained on quantum data, and semi-empirical quantum chemistry for CO$_2$ and the capped Ace--Phe--Tyr--NMe peptide. The method reproduces spectra within the chosen band, suppresses out-of-band response, reveals mode coupling, and demonstrates force-field dependence of spectral features, especially for the thermodynamically important low frequencies. 
FIMD offers an efficient and transparent way to probe the vibrational physics underlying spectroscopic and calorimetric observables.
\end{abstract}

\maketitle

Molecular vibrations shape stability, reactivity, and spectroscopic signatures. Yet in most workflows molecular dynamics (MD) is still integrated in the time domain \cite{rapaport2004art}, and frequency selectivity is introduced only \emph{after} the trajectory is generated, for example by analysing time-correlation functions and their Fourier transforms \cite{vitale2015anharmonic}. A natural question is whether one can instead perform \emph{symplectic} and \emph{time‑reversible} MD directly in a Fourier/modal representation, while keeping force evaluations tied to physically consistent real-space molecular geometries. Here we introduce Fourier--integrator molecular dynamics (FIMD), which answers this by formulating an exact harmonic drift in mode space and treating all remaining effects as explicit source terms; in the absence of stochastic sources, the resulting drift--kick update is symplectic.
Moreover, in the harmonic limit both the full and band-limited schemes preserve the canonical (Boltzmann) measure over a subspace and satisfy modewise equipartition, providing a principled basis for viewing band restriction as a controlled reduction of the active phase space rather than a post-processing filter.

The obvious use case for such an integrator is to generate infra red (IR) spectra, for example targeted calculations of the amide I-II bands often used empirically as reporters of protein structure. FIMD is also an efficient route to the finite temperature vibrational density of states (DOS) including anharmonic effects. The workflow of the proposed thermodynamic usecase would be to obtain multiple candidate configurational states for a biomolecule (for example using machine learning tools) and to run a band-limited FIMD simulation for each of them in order to evaluate from the DOS in the immediate vicinity around each of the candidate energy minimum structures. Low frequency vibrations contribute the dominant terms of the entropy, therefore band-limiting to low frequencies should capture the important thermodynamics per-configuration efficiently. 

(Quasi-)harmonic normal mode approaches provide compact representations around reference structures~\cite{wilson1980molecular,bahar2010normal}, and covariance-based frameworks quantify collective fluctuations and entropic contributions from trajectories~\cite{schlitter1993estimation,andricioaei2001calculation}. However, stable integration is dictated by the fastest vibrations, so uniform time stepping heavily samples high frequency motion while low-frequency, collective dynamics remain hard to sample over accessible trajectory lengths \cite{prasad2018best}. Isolating or manipulating a particular band when performed retrospectively is therefore cumbersome. In classical force fields the timestep limitation from the fastest X--H stretches is often alleviated by imposing holonomic bond constraints (e.g., SHAKE) or by hydrogen-mass repartitioning, whereas such interventions are problematic for quantum chemistry and machine-learned potentials---further motivating a frequency-selective propagation strategy~\cite{ryckaert1977shake,hopkins2015hmr}.

A substantial literature has developed towards extraction of rich frequency-resolved information from trajectories: power-spectrum estimators for anharmonic systems \cite{kleinhesselink1992evaluation,wang2020benzene}, signal-processing tools such as filter-diagonalisation to obtain vibrationally resolved spectra from short windows \cite{jungwirth1997vibrationally,da1998anharmonic}, and spectral energy density methods for mode-resolved lifetimes and transport \cite{thomas2010predicting}. Recent work has also moved toward explicitly frequency-selective collective coordinates and mode-coupling diagnostics: Sauer and Heyden introduced a fully anharmonic, frequency-dependent mode analysis (FRESEAN) based on Fourier-transformed velocity cross-correlations \cite{sauerHeyden2023fresean}, and information-theoretic measures have been used to quantify dependencies beyond linear correlation in mode-like representations \cite{numata2012entropy}. These approaches strengthen interpretation, but they do not change the propagator: frequency selectivity still enters only after the dynamics have been generated.

Complementary strategies attempt to \emph{influence} sampling in selected frequency ranges during time-domain propagation. Coloured-noise Langevin schemes tune the frequency-dependent thermostat response \cite{ceriotti2010coloured,rossi2014colored,brünig2022pair}, quantum thermal baths inject frequency-dependent fluctuations to mimic quantum populations \cite{dammak2009quantum}, and earlier ``Fourier acceleration'' ideas modify stochastic dynamics to reduce critical slowing down \cite{alexander2001fourier}. Separately, time-domain integrators can exploit local quadratic information to treat stiff motion analytically within a step \cite{millam1999hessian,bakken1999updating,lourderaj2007dd}. However, none of these provide a symplectic Fourier-space integrator that can strictly exclude chosen bands \emph{during} propagation. FIMD addresses this missing piece by making frequency selectivity part of the integrator: the propagation can be restricted to prescribed frequency windows \emph{during} the dynamics, with a strict modal projection that suppresses artificial out-of-band leakage, while a band-subspace Langevin thermostat provides controlled modewise thermalisation in clustered spectra.


\textit{Fourier--integrator molecular dynamics (FIMD).}\label{sec:theory}
FIMD advances molecular dynamics in a Fourier/modal representation, so that (i) the harmonic component of the Liouville flow becomes an exact phase advance, and (ii) hard band limitation can be imposed as a spectral truncation at the level of the generators.
Let $\mathbf r\in\mathbb R^{3N}$, $\mathbf p\in\mathbb R^{3N}$, diagonal mass matrix $M$, and Hamiltonian $H(\mathbf r,\mathbf p)=\tfrac12 \mathbf p^\mathsf T M^{-1}\mathbf p + V(\mathbf r)$.
The Liouville operator acting on phase-space observables $f$ is split into positional and momentum parts,
\begin{equation}
\hat{\mathcal L}f
=
\underbrace{(M^{-1}\mathbf p)\cdot\nabla_{\mathbf r} f}_{\hat{\mathcal L}_r f}
+
\underbrace{\mathbf F(\mathbf r)\cdot\nabla_{\mathbf p} f}_{\hat{\mathcal L}_p f},
\label{eq:liouville_split}
\end{equation}
where the force vector is obtained via $\mathbf F(\mathbf r)=-\nabla_{\mathbf r}V(\mathbf r)$.
A second-order Strang splitting yields velocity Verlet.
The key point for FIMD is that $\hat{\mathcal L}_r$ and $\hat{\mathcal L}_p$ generate translations; in the Fourier--Liouville basis used for the propagator, their exponentials act diagonally on Fourier coefficients.
Explicit Fourier expanded generators and short-time maps are given in the End Matter, and hard band limitation is implemented by truncating the retained Fourier indices.

If a Hessian is conveniently available, a local harmonic reference can be defined at geometry $\mathbf r_0$ by the Hessian $H_0=\nabla\nabla V(\mathbf r_0)$ and the mass-weighted displacement $\mathbf x=M^{1/2}(\mathbf r-\mathbf r_0)$.
Diagonalisation of the mass-weighted Hessian gives
\begin{equation}
M^{-1/2}H_0M^{-1/2}=W\,\Omega^2 W^\mathsf T,
\label{eq:hessian_diag}
\end{equation}
with $\Omega=\mathrm{diag}(\omega_1,\ldots,\omega_{n^\nu})$, $n^\nu \in \mathbb{N}$, where $W$ contains vibrational eigenvectors (with translations/rotations removed).
Modal coordinates and conjugate momenta are
\begin{equation}
\mathbf q=W^\mathsf T\mathbf x,
\qquad
\boldsymbol\pi=W^\mathsf T M^{-1/2}\mathbf p .
\label{eq:qpi_def}
\end{equation}
In this case we use the effective harmonic operator implied by $(W,\Omega)$, i.e.,
$M^{-1/2}H_0M^{-1/2}=W\Omega^2W^\mathsf T$.

When a Hessian is unavailable (or undesired) at finite temperature, the reference $(\mathbf r_0,W,\Omega)$ can be obtained from a short equilibrated trajectory (NVE or weakly thermostatted NVT) segment long enough to resolve the target band ($\delta\omega\sim 2\pi/T_{\mathrm{ref}}$).
After removing overall translation and rotation, we set $\mathbf r_0=\langle\mathbf r\rangle$ and form $\mathbf x(t)=M^{1/2}(\mathbf r(t)-\mathbf r_0)$; from the calibration segment we form a windowed covariance estimate $\widehat C\approx C=\langle \mathbf x\mathbf x^\mathsf T\rangle$ and estimate the columns of $W$ as an orthonormal basis spanning the dominant eigenspace of $\widehat C$. Frequency estimates directly from $\widehat C$ eigenvalues have a known systematic bias \cite{dryden} so here frequencies $\omega_\nu$ are assigned from Fourier spectra (or velocity autocorrelation spectra) of the projected velocities $v_\nu(t)=\mathbf w_\nu^\mathsf T M^{-1/2}\mathbf p(t)$.

\begin{figure*}[t]
  \centering
  \includegraphics[width=\linewidth]{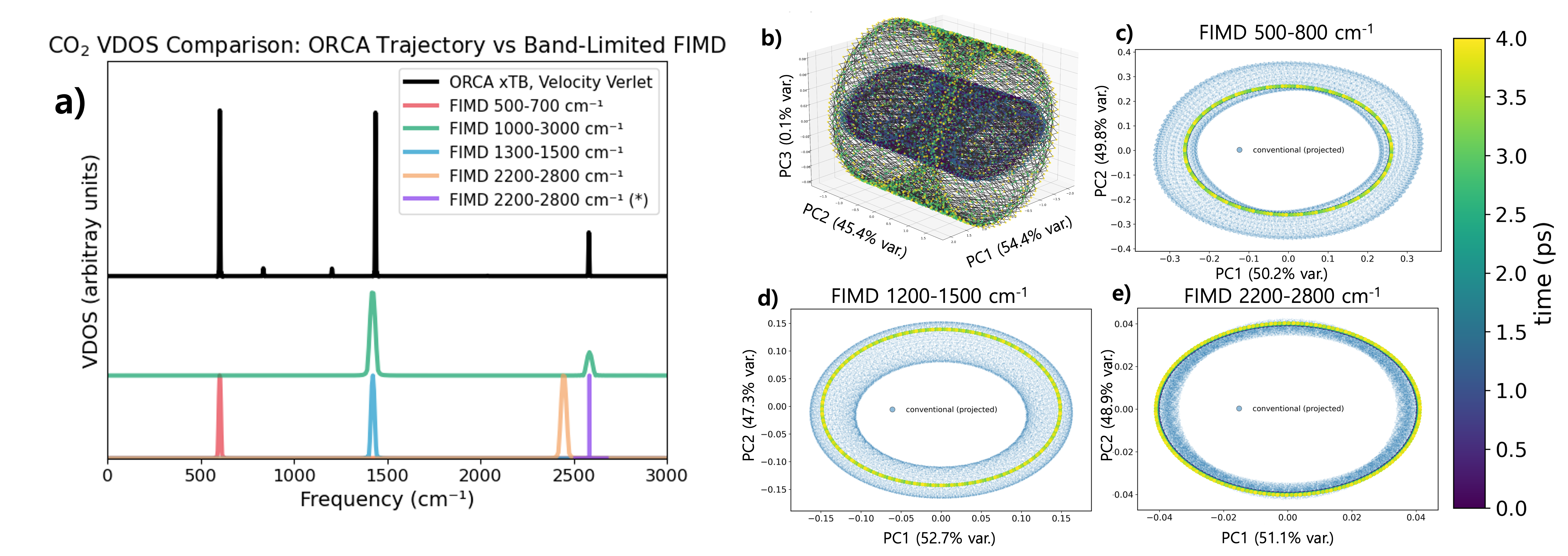}
\caption{\textbf{CO$_2$ as a minimal demonstrator of band-limited FIMD and mode-resolved phase-space structure.}
(a) Vibrational density of states (VDOS) from a conventional ORCA xTB velocity Verlet trajectory (black) compared with band-limited FIMD runs (coloured), where only modes with reference frequencies $\omega_\nu\!\in[\omega_{\min},\omega_{\max}]$ are propagated and forces are evaluated on the band-limited reconstruction $\mathbf r_B$; the chosen windows isolate bending/stretching features and suppress off-band response.
For 2200--2800~cm$^{-1}$, we compare a Hessian-based modal basis at the equilibrated geometry (2200--2800) with an alternative basis obtained from a trajectory via a spectral/Fourier construction (2200--2800$^{*}$), the latter yielding a peak closer to the conventional reference.
(b) Three-dimensional PCA of the \emph{three-mode} FIMD phase-space trajectory, showing a decomposition into three quasi-cylindrical components associated with the coexisting mode motions.
(c)--(e) For single-band FIMD runs (500--800, 1200--1500, and 2200--2800~cm$^{-1}$), the band phase space projected onto the first two principal components (PCs) yields near-elliptic annuli (coloured by time, 0--4~ps), while the conventional velocity Verlet trajectory projected onto the \emph{same band-specific} PC bases (blue) forms a thicker, more exploratory cloud that visualises additional coupling and out-of-band influence beyond the idealised band-limited rotation.}

  \label{fig:co2}
\end{figure*}

For a chosen active set of vibrational frequencies $B$, define the modal projector $P_B$ (diagonal in the mode index $\nu$), and set $\mathbf q_B=P_B\mathbf q$, $\boldsymbol\pi_B=P_B\boldsymbol\pi$.
In these variables the harmonic Hamiltonian
$\tfrac12\sum_\nu(\pi_\nu^2+\omega_\nu^2 q_\nu^2)$ produces an exact phase advance.
The time-$\Delta t$ harmonic map is written componentwise as
\begin{align}
q_\nu(t+\Delta t)
&=
q_\nu(t)\cos(\omega_\nu\Delta t)
+
\frac{\pi_\nu(t)}{\omega_\nu}\sin(\omega_\nu\Delta t),
\label{eq:harm_q}\\
\pi_\nu(t+\Delta t)
&=
\pi_\nu(t)\cos(\omega_\nu\Delta t)
-
\omega_\nu q_\nu(t)\sin(\omega_\nu\Delta t).
\label{eq:harm_pi}
\end{align}
This closed-form oscillatory update is the modal analogue of the drift flow in the split propagator.

The full force is decomposed into the harmonic reference about $\mathbf r_0$ plus a residual,
\begin{align}\label{eq:force_split}
\mathbf F(\mathbf r)
&=
-\;H_0(\mathbf r-\mathbf r_0)
+
\mathbf F_\Delta(\mathbf r).
\notag
\end{align}
In practice, $\mathbf F_\Delta(\mathbf r)$ is formed from a single Cartesian force evaluation $\mathbf F(\mathbf r)=-\nabla V(\mathbf r)$ by subtracting the harmonic baseline force $-H_0(\mathbf r-\mathbf r_0)$.

\textit{Band selection and band-limited propagation.}
A short equilibrated trajectory segment is used to estimate the vibrational spectrum and select a target window $[\omega_{\min},\omega_{\max}]$, whose upper cutoff sets a natural timestep scale for band-limited propagation, $\Delta t \lesssim \pi/\omega_{\max}$ (Nyquist).
The active band is then defined by the reference frequencies $\omega_\nu\in[\omega_{\min},\omega_{\max}]$,
where $\{\omega_\nu\}$ arise either from the Hessian diagonalisation in Eq.~\eqref{eq:hessian_diag} or from the
trajectory-based finite-temperature reference described above, defining
\begin{equation}
B=\{\nu:\ \omega_{\min}\le \omega_\nu\le \omega_{\max}\}.
\label{eq:band_def}
\end{equation}
Let $P_B$ denote the projector in mode-index space onto $\nu\in B$, and let $W_B$ be the restriction of $W$
to these modes.
Force evaluation is performed on the band-limited reconstruction
\begin{equation}
\mathbf r_B=\mathbf r_0+M^{-1/2}W_B\,\mathbf q_B.
\label{eq:reconstruct_rB}
\end{equation}
At each step, the Cartesian force is evaluated at $\mathbf r_B$ and mapped to modal space by
\begin{equation}
\mathbf s(\mathbf r_B)=W^\mathsf T M^{-1/2}\mathbf F_\Delta(\mathbf r_B),
\label{eq:source_band}
\end{equation}
after which only $P_B\mathbf s$ updates the active momenta.
A single deterministic band-limited step is the symmetric composition
\begin{align}
\boldsymbol\pi_B^{n+\frac12}
&=
\boldsymbol\pi_B^{n}
+
\frac{\Delta t}{2}\,P_B\,\mathbf s(\mathbf r_B^{\,n}),
\label{eq:fimd_kick1}\\
(\mathbf q_B^{n+1},\boldsymbol\pi_B^{n+\frac12,\mathrm{drift}})
&=
\Phi_{0,B}^{\Delta t}\!\left(\mathbf q_B^{n},\boldsymbol\pi_B^{n+\frac12}\right),
\label{eq:fimd_harm}\\
\boldsymbol\pi_B^{n+1}
&=
\boldsymbol\pi_B^{n+\frac12,\mathrm{drift}}
+
\frac{\Delta t}{2}\,P_B\,\mathbf s(\mathbf r_B^{\,n+1}),
\label{eq:fimd_kick2}
\end{align}
where $\Phi_{0,B}^{\Delta t}$ applies Eqs.~\eqref{eq:harm_q}--\eqref{eq:harm_pi} independently for each
$\nu\in B$. Modes outside $B$ are held fixed, implementing a hard spectral truncation in the modal/Fourier picture.

The reference $(\mathbf r_0,H_0,W,\Omega)$ is refreshed according to a geometry-based criterion,
which may update the active set $B$.
For fixed reference and without thermostatting, Eqs.~\eqref{eq:fimd_kick1}--\eqref{eq:fimd_kick2} define a symmetric second-order, time-reversible map.

For fixed reference and in the microcanonical setting, the band update is \emph{symplectic}.
It is a symmetric kick--drift--kick (Strang) splitting of the band Hamiltonian into the exactly solvable quadratic reference flow $\Phi_{0,B}^{\Delta t}$ and a potential kick generated by the conservative remainder force $\mathbf F_\Delta$ (restricted to $B$).
Since each substep is an exact Hamiltonian flow, their symmetric composition is a second-order symplectic map on the band phase space.

If finite-temperature sampling is employed, a thermostat may be applied only to the active band variables, leaving the complementary modes ($\nu\notin B$) unchanged.
A convenient choice is an Ornstein--Uhlenbeck update of each band momentum $\pi_\nu$ ($\nu\in B$) with friction, which draws the correct Maxwell--Boltzmann marginal for $\pi_B$ at temperature $T$.
When composed with the deterministic band propagation in the usual splitting fashion, this yields an NVT scheme that controls temperature without altering the inactive subspace.

Since the harmonic drift is evaluated in closed form, discretisation effects enter only through the symmetric splitting with the residual-force kick, giving a second-order global error in $\Delta t$; if $\Delta t$ is too large, the residual forcing is under-resolved and can manifest as small phase/amplitude errors and corresponding peak shifts/broadening in in-band spectra.
When frequencies are expressed as wavenumbers $\tilde\nu$ (cm$^{-1}$), the basic sampling bound is $\Delta t < [2c\,\tilde\nu_{\max}]^{-1}$.
Finally, the window $[\tilde\nu_{\min},\tilde\nu_{\max}]$ is a modal \emph{selection} and should be distinguished from physical linewidth: anharmonic broadening and inter-band mixing are governed by the residual force $\mathbf F_\Delta(\mathbf r)$, which vanishes for a purely quadratic potential; stronger residuals imply greater transfer to sidebands/out-of-band features and may therefore require wider windows (or more frequent reference updates) to reproduce coupled response.

Stationarity, mode-wise equipartition, preservation of the canonical (Boltzmann) measure, and additional NVE/NVT validation are provided in \cite{FIMD_SM}.

\textit{CO$_2$ as a prototypical example.}
\begin{figure*}[t]
\centering
\includegraphics[width=\textwidth]{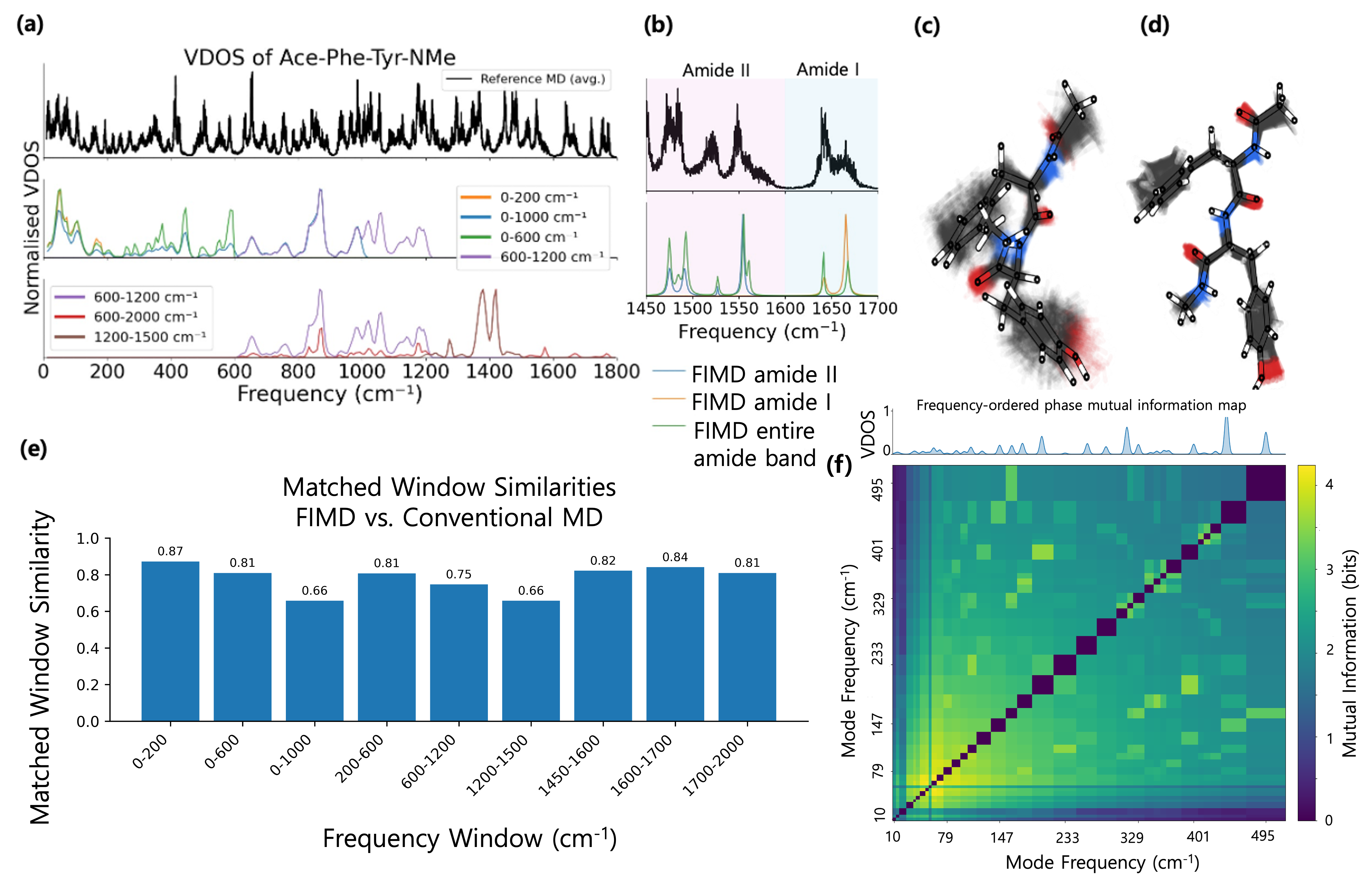}
\label{fig:fig2}
\caption{\textbf{SO3LR band-limited FIMD reproduction of peptide VDOS and phase-coupling diagnostics.}
(a) Velocity density of states (VDOS) of Ace--Phe--Tyr--NMe obtained from a conventional Cartesian velocity--Verlet trajectory (black) and from band-limited FIMD propagations (coloured) using the SO3LR potential with integration settings as in the previous figure; FIMD traces at various frequency windows are normalised for clarity.
(b) Zoom-in of the amide region (1450--1700~cm$^{-1}$) comparing the reference MD VDOS (black) with FIMD restricted to Amide~II (1450--1600~cm$^{-1}$), Amide~I (1600--1700~cm$^{-1}$), and the full amide band (1450--1700~cm$^{-1}$).
(c,d) Trajectory overlays for the reference MD run (c) and the 0--200~cm$^{-1}$ band-limited FIMD run (d), illustrating that even the low-band dynamics do not cross over to diffusive, exploratory character.
(e) Matched-window spectral similarity between the conventional SO3LR MD VDOS and the corresponding band-limited FIMD VDOS for each propagated frequency window.
Bar heights show the raw mass-weighted Jensen--Shannon similarity score \(S\) for the matched band/window pair,
with \(S=1\) corresponding to identical windowed spectra with matching spectral weight.
(f) Frequency-ordered phase mutual-information (MI) map for the SO3LR trajectory,
with the corresponding VDOS shown above on the same physical frequency axis.
The VDOS strip was obtained from the detrended and Lorentzian-decomposed spectrum using Tihi~\cite{HanTihi2024}.
Each off-diagonal cell reports phase MI in bits for a pair of Fourier-resolved mode phases; details of the phase-MI construction are given in \cite{FIMD_SM}
The pronounced low-frequency block and extended low-to-mid-frequency structure indicate organised soft-mode phase dependence.}
\end{figure*}

A short, 10 ps, conventional velocity Verlet segment defines the reference $(\mathbf r_0,W,\Omega)$ used for band selection and for the band-limited reconstruction $\mathbf r_B$, after removing overall translation/rotation. Fig.~\ref{fig:co2}(a) compares the VDOS from the conventional xTB trajectory (black) with band-limited FIMD runs (coloured): narrow windows isolate the expected bending and stretching features, while broader windows recover the corresponding spectral content within a controlled frequency range. For the 2200--2800~cm$^{-1}$ window we additionally compare two choices of the modal/Fourier basis: one obtained from a Hessian at the global minimum geometry, and one obtained directly from an existing trajectory via a Fourier/spectral construction (marked with \,*). The trajectory-derived basis yields a high-frequency peak that aligns more closely with the conventional velocity Verlet reference, consistent with capturing the finite-temperature effective frequency content relevant for the subsequent band-limited propagation.

Beyond spectral isolation, CO$_2$ provides a geometric diagnostic of band-resolved dynamics in phase space. Performing Principal Components Analysis (PCA) as a convenient dimensionality reduction tool on the \emph{three-mode} FIMD trajectory (Fig.~\ref{fig:co2}(b)) reveals that the motion is organised into three distinct cylindrical components, reflecting the coexistence of the three vibrational-mode motions within the full trajectory. When FIMD is restricted to a single vibrational band, the corresponding band phase space collapses to a near-elliptic annulus in the first two PCs (Fig.~\ref{fig:co2}(c--e)). This is the expected signature of predominantly harmonic rotation under the quadratic phase advance, with deviations introduced by the projected residual (anharmonic) kicks. Projecting the conventional velocity Verlet trajectory onto these \emph{band-specific} PC bases (blue points in Fig.~\ref{fig:co2}(c--e)) typically yields a thicker cloud around the annulus, providing direct evidence for additional coupling and out-of-band influence beyond the band-limited propagation.

Figure~\ref{fig:co2} shows that, in a trajectory-learned vibrational basis, band-limited FIMD recasts molecular motion into mode resolved action–angle mechanics, expressed by isolated spectral features and near-elliptic phase space annuli. Anharmonic coupling (and possible chaos) then appears as distortions of these orbits, providing a transferable diagnostic for larger molecules.

\textit{Band-limited FIMD validation for Ace--Phe--Tyr--NMe peptide with SO3LR.}
Figure~\ref{fig:fig2}(a) compares the velocity density of states (VDOS) of a conventional reference MD trajectory against multiple band-limited FIMD trajectories, all propagated with the same machine-learned SO3LR force field trained on DFT (PBE0+MBD) data~\cite{kabylda2025molecular}.
All runs use the same timestep and trajectory duration as in the previous figure.
Across all windows, the FIMD spectra retain the in-band structure while strongly suppressing the out-of-band response, providing a controlled way to isolate the vibrational content of a prescribed frequency interval; the matched-window similarity scores in Fig.~2(e) quantify this recovery.

Figure~\ref{fig:fig2}(b) zooms into the amide region (1450--1700~cm$^{-1}$).
The reference MD VDOS (black) is shown together with band-limited FIMD restricted to Amide~II (1450--1600~cm$^{-1}$), Amide~I (1600--1700~cm$^{-1}$), and the full amide band (1450--1700~cm$^{-1}$).
This decomposition resolves the amide substructure and serves as a stringent test of how spectral intensity is redistributed by coupling within the broader amide manifold.
While the Amide~I-only and Amide~II-only propagations reproduce the correct peak \emph{positions}, they can underrepresent certain features (missing peaks or distorted amplitudes) because intensity borrowing between the two sub-bands is suppressed by construction.
Including the entire amide window restores these cross-band couplings and correspondingly recovers the missing intensity, indicating that simulations targeted to this empirically important region of the spectrum are practical using FIMD providing that the whole amide band is covered.

Figure~\ref{fig:fig2} (c,d) visualise the resulting configurational exploration using trajectory overlays.
The conventional MD trajectory (c) spans a broad region of configuration space in this projection, whereas the 0--200~cm$^{-1}$ band-limited FIMD trajectory (d) occupies a markedly smaller subset: periodic re-basing and re-centring of the vibrational basis space should eventually permit a random walk over the configurational space of the molecule, however the FIMD algorithm is targeted to efficient investigation of vibrations, not to landscape exploration.

Panel~(e) quantifies the same spectral recovery more systematically across all propagated SO3LR frequency windows.
For each matched band/window pair, the conventional MD and band-limited FIMD VDOS were restricted to the same frequency interval, interpolated on a common grid, normalised within that interval, and compared using the
mass-weighted Jensen--Shannon similarity score \(S\) defined in the Supplementary Materials.
The matched-window scores are high for most windows.
Thus, beyond the specific Amide examples in panel~(b), FIMD generally reproduces the spectra across frequency bands.

The timestep-stability scan for the \(0\!-\!200~\mathrm{cm}^{-1}\)
SO3LR run is reported in \cite{FIMD_SM}; in the present implementation the dynamics remain bounded up to \(\Delta t=4.0~\mathrm{fs}\) and become unstable at \(\Delta t=4.5~\mathrm{fs}\).
A Nyquist argument for a strictly band-limited signal with cutoff $\tilde\nu_{\max}$ gives $\Delta t_{\mathrm{Nyq}}\approx(2c\,\tilde\nu_{\max})^{-1}$, i.e.\ $\Delta t_{\mathrm{Nyq}}\approx 83$~fs for $\tilde\nu_{\max}=200$~cm$^{-1}$.
This is only a heuristic motivation: sampling bounds do not control numerical stability or integration error.
For velocity Verlet (St\"ormer--Verlet), even the idealised harmonic-mode analysis imposes a much stricter constraint, and symplectic backward-error theory explains bounded long-time energy error only when the step is chosen conservatively~\cite{hairer2003geometric,toxvaerd2012energy}.
Moreover, coupled multi-mode dynamics can exhibit additional nonlinear resonance instabilities in splitting/MTS-type schemes, producing sharp timestep barriers not predictable from sampling considerations alone~\cite{ma2003verlet,minary2004long,leimkuhler2004simulating}.

Panel~(f) provides a complementary phase-coupling diagnostic for the SO3LR
trajectory. The heat map reports the mutual information
$I(\theta_i;\theta_j)$ between instantaneous phases of Fourier-resolved
vibrational modes, ordered by their frequency centres; the VDOS shown above
the map relates these phase dependencies to the spectral intensity. The
frequency centres used for this mode ordering were obtained from the
detrended, Lorentzian-decomposed VDOS using Tihi~\cite{HanTihi2024}; the full
phase-MI construction is described in \cite{FIMD_SM}. 
The pronounced low-frequency block and extended low-to-mid-frequency structure
indicate statistically organised phase dependence among soft modes, consistent
with collective anharmonic coupling rather than independent harmonic motion.
The corresponding AMBER ff14SB and ORCA xTB phase-MI maps are reported in \cite{FIMD_SM}; together, they show that the low-frequency
phase-coupling structure is strongly force-field dependent, with SO3LR giving
the most clearly organised soft-mode phase structure in the present analysis.

The full band--window Jensen--Shannon similarity matrices for SO3LR, ORCA xTB, and AMBER ff14SB are reported in \cite{FIMD_SM}.
For SO3LR, the matched band/window entries dominate the corresponding off-window entries, with an average matched score of
\(\langle S\rangle_{\rm match}=0.74\) compared with
\(\langle S\rangle_{\rm mismatch}=0.089\).  Panel~2(e) shows the matched SO3LR entries directly, making clear which spectral regions are well reproduced and which remain more challenging under the present band-limited propagation protocol.

\textit{Discussion and Outlook.}
We introduced Fourier--integrator molecular dynamics (FIMD), a trajectory-level formulation of classical MD in a vibrational/modal basis.
FIMD is constructed by a symmetric splitting between an exactly solvable quadratic reference flow and an explicit residual force term (including thermostatting when used), yielding a second-order, structure-preserving update.
Limitations arise when the local quadratic reference degrades under strong anharmonicity or rapid frequency drift, motivating more frequent updates or multi-reference variants.
Looking ahead, FIMD should enable targeted studies of complex vibrational manifolds---including strongly coupled bands and condensed-phase settings---by turning band selection into a practical knob for tuning into inter-mode energy flow.
Future work will develop adaptive band selection and reference-update schemes that preserve structure while extending applicability to more anharmonic regimes, look for memory-efficient bootstrapping methods, and investigate the potential of frequency controlled molecular dynamics as an accelerated sampling scheme.

\textit{Acknowledgments}
The authors thank Miguel Gallegos, Florian Brünig, Guido Di Paola, Giuseppe Mansi and Igor Poltavskyi of the University of Luxembourg and Gyunam Park of Chung-Ang University for helpful readings and discussions; and F.~Simone Ruggieri of the University of Wageningen for grounding in the empirical side of the state of the art. We acknowledge the grant C20/MS/14588607 ``QUantum Infra-Red Efficiently, QUIRE'' of the \emph{Fonds Nationale de la Recherche} (FNR) Luxembourg.

\bibliography{apssamp}

\pagebreak
\appendix

\providecommand{\Lopr}{\hat{\mathcal L}_r}
\providecommand{\Lopp}{\hat{\mathcal L}_p}

\section*{Fourier--expanded Liouville generators and the harmonic phase advance.}
This section records the Fourier--Liouville identities that underlie the band-limited propagators cited in the main text.

\textit{Fourier representation of a degree of freedom.}
For each Cartesian degree of freedom $\ell$, a truncated Fourier representation is introduced,
\begin{equation}
  r^{(\ell)}(t)
  = \sum_{k} r^{(\ell)}_{k}\,e^{i\omega^{(\ell)}_{k} t},
  \label{eq:fourier_r}
\end{equation}
where the index set for $k$ specifies the retained frequencies.
Band limitation corresponds to restricting this index set to frequencies within a chosen window.

\textit{Forcing in Fourier amplitudes.}
Let $L$ denote the Lagrangian associated with the dynamics. The functional derivative with respect to
$r^{(\ell)}(t)$ admits the Fourier-amplitude expansion
\begin{equation}
  \frac{\delta L}{\delta r^{(\ell)}(t)}
  = \sum_{k}
    e^{-i\omega^{(\ell)}_{k} t}\,
    \frac{\partial L}{\partial r^{(\ell)}_{k}},
  \label{eq:deltaL_fourier}
\end{equation}
and the corresponding momentum equation is
\begin{equation}
  \dot p^{(\ell)}(t)
  = \sum_{k}
    e^{-i\omega^{(\ell)}_{k} t}\,
    \frac{\partial L}{\partial r^{(\ell)}_{k}}.
  \label{eq:pdot_fourier}
\end{equation}

\textit{Fourier-expanded split generators.}
Consider the one-degree-of-freedom split generators (associated with the $\hat{\mathcal L}_r$ and
$\hat{\mathcal L}_p$ in the main text) acting on an observable $f$.
Using \eqref{eq:fourier_r} and the chain rule through the Fourier amplitudes gives the positional generator
\begin{align}
  \Lopr^{(\ell)} f
  &=
  \sum_{k,j}
  e^{i(\omega^{(\ell)}_{k}-\omega^{(\ell)}_{j})t}\,
  \Bigl(\dot r^{(\ell)}_{k}
        + i\omega^{(\ell)}_{k} r^{(\ell)}_{k}\Bigr)\,
  \frac{\partial f}{\partial r^{(\ell)}_{j}},
  \label{eq:Lr_fourier}
\end{align}
and inserting \eqref{eq:pdot_fourier} yields the momentum generator
\begin{align}
  \Lopp^{(\ell)} f
  &=
  \sum_{k,j}
  e^{-i\omega^{(\ell)}_{k} t}\,
  \frac{\partial r^{(\ell)}_{j}}{\partial p^{(\ell)}_{k}}\,
  \frac{\partial L}{\partial r^{(\ell)}_{k}}\,
  \frac{\partial f}{\partial r^{(\ell)}_{j}}.
  \label{eq:Lp_fourier}
\end{align}
Equations \eqref{eq:Lr_fourier}--\eqref{eq:Lp_fourier} show explicitly that spectral truncation (restriction
of the Fourier indices $k$) induces a hard band limitation at the operator level.

\textit{Short-time propagators.}
Introduce the compact differential operator
\begin{equation}
\mathcal D^{(\ell,\Delta t)}_{j,m}
:=
\frac{(\Delta t)^m}{m!}\,
\frac{\partial^m}{\partial (r^{(\ell)}_{j})^m}.
\label{eq:Djm_def}
\end{equation}
A Taylor expansion of the exponential maps generated by \eqref{eq:Lr_fourier}--\eqref{eq:Lp_fourier} then gives
\begin{align}
  e^{\Delta t\,\Lopr^{(\ell)}} f
  &=
  \sum_{m=0}^{\infty}\sum_{k,j}
  e^{\,im(\omega^{(\ell)}_{k}-\omega^{(\ell)}_{j})\Delta t}\,
  \mathcal{K}_r^{mk,\ell}\,
  \mathcal D^{(\ell,\Delta t)}_{j,m} f
  \label{eq:expLr_fourier}\\[2pt]
  e^{(\Delta t/2)\,\Lopp^{(\ell)}} f
  &=
  \sum_{m=0}^{\infty}\sum_{k,j}
  e^{-\,\tfrac{i}{2}m\omega^{(\ell)}_{k}\Delta t}\,
  \mathcal{K}^{mk,\ell}_p
  \mathcal D^{(\ell,\Delta t)}_{j,m} f.
  \label{eq:expLp_fourier}
\end{align}
with
\begin{align}
\mathcal{K}_r^{mk,\ell}&=\Bigl(\dot r^{(\ell)}_{k}+ i\omega^{(\ell)}_{k} r^{(\ell)}_{k}\Bigr)^{m}, \\
    \mathcal{K}_p^{mk,\ell}&=\left( \frac{1}{2}\frac{\partial r^{(\ell)}_{j}}{\partial p^{(\ell)}_{k}}\,\frac{\partial L}{\partial r^{(\ell)}_{k}}\right)^{m}. \nonumber
\end{align}
In band-limited dynamics, the sums over $k$ are restricted to the retained Fourier indices; this restriction
is the algebraic expression of hard spectral truncation.

\textit{Harmonic phase advance and recovery of cosine--sine updates.}
In the harmonic limit, each normal-mode coordinate contains only a two-sided single-frequency pair.
For one mode $\nu$ with frequency $\omega_\nu$, write
\begin{equation}
q_\nu(t)=a_\nu e^{i\omega_\nu t}+a_\nu^{*}e^{-i\omega_\nu t}.
\label{eq:q_two_sided}
\end{equation}
A time shift by $\Delta t$ multiplies the Fourier components by $e^{\pm i\omega_\nu\Delta t}$, implying
\begin{equation}
q_\nu(t+\Delta t)
=
q_\nu(t)\cos(\omega_\nu\Delta t)
+\frac{\dot q_\nu(t)}{\omega_\nu}\sin(\omega_\nu\Delta t),
\label{eq:q_shift_cos_sin}
\end{equation}
and, by differentiation,
\begin{equation}
\dot q_\nu(t+\Delta t)
=
\dot q_\nu(t)\cos(\omega_\nu\Delta t)
-\omega_\nu q_\nu(t)\sin(\omega_\nu\Delta t).
\label{eq:qdot_shift_cos_sin}
\end{equation}
Upon identifying $\pi_\nu=\dot q_\nu$ in mass-weighted modal variables, these relations recover the exact
harmonic phase-advance map used for the drift step in the main text: the harmonic generator is diagonal in a
Fourier/modal basis, and its exponential is a pure phase multiplication.

\textit{Band-subspace thermostat.}
If finite-temperature sampling is employed, a thermostat may be applied only to the active band variables.
For example, an Ornstein--Uhlenbeck update for each band momentum $\pi_\nu$ ($\nu\in B$) with friction $\gamma$ is
\begin{equation}
\pi_\nu \leftarrow e^{-\gamma\Delta t}\pi_\nu
+\sqrt{(1-e^{-2\gamma\Delta t})\,k_B T}\;\xi_\nu,
\qquad
\xi_\nu\sim\mathcal N(0,1),
\label{eq:OU_band}
\end{equation}
while $\nu\notin B$ remains unchanged.


\section*{Symplecticity of the band-limited Fourier integrator}

\begin{theorem}[Symplecticity of fixed-reference band-limited FIMD]
\label{thm:supp_symplecticity}
Assume the reference $(\mathbf r_0,H_0,W,\Omega)$ and active set $B$ are held fixed over a step.
Let $P_B$ be the projector onto the retained mode indices, and define the band reconstruction
\begin{equation}
  \mathbf r_B(\mathbf q_B)=\mathbf r_0 + M^{-1/2}W_B\,\mathbf q_B,
  \label{eq:supp_reconstruct_rB}
\end{equation}
with $W_B$ the restriction of $W$ to columns in $B$.
Define the residual potential and residual force by
\begin{align}\label{eq:supp_DeltaV_def}
  \Delta V(\mathbf r)=V(\mathbf r)-V(\mathbf r_0)-\tfrac12(\mathbf r-\mathbf r_0)^{\mathsf T}H_0(\mathbf r-\mathbf r_0),\\
  \mathbf F_\Delta(\mathbf r)=-\nabla\Delta V(\mathbf r), \nonumber
\end{align}
and define the (projected) modal source
\begin{equation}
  \mathbf s_B(\mathbf r):=W_B^{\mathsf T}M^{-1/2}\,\mathbf F_\Delta(\mathbf r).
  \label{eq:supp_source_def}
\end{equation}
Let $\Phi_{0,B}^{\Delta t}$ denote the exact time-$\Delta t$ harmonic flow generated by
$H_{0,B}=\tfrac12\sum_{\nu\in B}(\pi_\nu^2+\omega_\nu^2 q_\nu^2)$.
Consider the kick--drift--kick map
\begin{align}\label{eq:supp_kdk_map}
\Psi_B^{\Delta t}
&:=K_{\Delta t/2}\circ\Phi_{0,B}^{\Delta t}\circ K_{\Delta t/2},\\
K_{\tau}&:(\mathbf q_B,\boldsymbol\pi_B)\mapsto\bigl(\mathbf q_B,\boldsymbol\pi_B+\tau\,\mathbf s_B(\mathbf r_B(\mathbf q_B))\bigr).\nonumber
\end{align}
Then $\Psi_B^{\Delta t}$ is symplectic on the band phase space $(\mathbf q_B,\boldsymbol\pi_B)$ with canonical form
$\omega_B=\sum_{\nu\in B} dq_\nu\wedge d\pi_\nu$.
\end{theorem}

\begin{proof}
Define the band Hamiltonian
\begin{align}
H_B(\mathbf q_B,\boldsymbol\pi_B)&=H_{0,B}(\mathbf q_B,\boldsymbol\pi_B)+H_{\Delta,B}(\mathbf q_B),
\\
H_{\Delta,B}(\mathbf q_B)&:=\Delta V\bigl(\mathbf r_B(\mathbf q_B)\bigr).
\end{align}

(i) $\Phi_{0,B}^{\Delta t}$ is the exact Hamiltonian flow of $H_{0,B}$, hence symplectic.

(ii) For $H_{\Delta,B}$, Hamilton's equations give $\dot{\mathbf q}_B=\nabla_{\boldsymbol\pi_B}H_{\Delta,B}=0$ and
$\dot{\boldsymbol\pi}_B=-\nabla_{\mathbf q_B}H_{\Delta,B}$.
Using \eqref{eq:supp_reconstruct_rB} and the chain rule,
\begin{align}
\nabla_{\mathbf q_B}H_{\Delta,B}(\mathbf q_B)
&=(M^{-1/2}W_B)^{\mathsf T}\nabla_{\mathbf r}\Delta V(\mathbf r_B)\\
&=-(W_B^{\mathsf T}M^{-1/2})\,\mathbf F_\Delta(\mathbf r_B)\\
&=-\mathbf s_B(\mathbf r_B),
\end{align}
so $\dot{\boldsymbol\pi}_B=\mathbf s_B(\mathbf r_B(\mathbf q_B))$.
Therefore the exact time-$\tau$ flow of $H_{\Delta,B}$ is precisely the shear $K_{\tau}$ in \eqref{eq:supp_kdk_map}, and hence is symplectic.

(iii) The full step $\Psi_B^{\Delta t}$ is the Strang composition of symplectic maps in \eqref{eq:supp_kdk_map}, so it is symplectic.
\end{proof}

\pagebreak
%

\clearpage
\onecolumngrid

\setcounter{section}{0}
\setcounter{figure}{0}
\setcounter{table}{0}
\setcounter{equation}{0}
\renewcommand{\thesection}{S\arabic{section}}
\renewcommand{\thefigure}{S\arabic{figure}}
\renewcommand{\thetable}{S\arabic{table}}
\renewcommand{\theequation}{S\arabic{equation}}

\begin{center}
  {\large\bfseries Supplementary Material for\\[2pt]
   ``Symplectic and Thermodynamically Consistent Molecular Dynamics
     in the Frequency Domain''}\\[6pt]
  {\normalsize Kyunghoon Han,\enspace Alexandre Tkatchenko,\enspace
   Joshua T.\ Berryman}\\[2pt]
  {\small\itshape Department of Physics and Materials Science,
   University of Luxembourg, Grand-Duchy of Luxembourg}
\end{center}

\vspace{1em}

\section*{Vibrational density of states comparison (AMBER FF14SB, ORCA xTB, SO3LR)}

Figure~\ref{fig:vdos_three_panel} presents the full stacked VDOS comparison for Ace--Phe--Tyr--NMe across AMBER ff14SB, ORCA xTB, and SO3LR.
In the main text, the spectral discussion is focused by showing a representative stacked panel (SO3LR) together with trajectory-level diagnostics (PCA portraits and Amide~I/II proxies), while the complete multi-model spectral context is provided here.
Within each force model, the band-limited FIMD traces recover the dominant peak positions and intra-band substructure \emph{inside} the propagated window and strongly suppress response outside it, consistent with hard modal projection during propagation.
Across models, the low-frequency region ($<600~\mathrm{cm^{-1}}$) is particularly discriminating: the stacked spectra show a progression from comparatively diffuse low-band content in AMBER, through more distinct low-band features in xTB, to sharper and more persistent low-frequency substructure in SO3LR, motivating the low-band trajectory portraits discussed in the main text.

Beyond confirming in-band recovery, the stacked spectra make inter-band contributions to selected peaks visually explicit, with the same qualitative behaviour observed across all three force models.
When a propagation window is narrowed, peaks that are prominent in the full Cartesian VDOS but weaken or disappear in the corresponding band-limited trace indicate features that are not sustained by dynamics confined to that band alone, but instead draw intensity from coupling to modes outside the retained subspace (e.g.\ combination-band character or anharmonic mixing). 
This effect is clearest in reciprocal-window comparisons: a feature that is reduced in a high-only window but re-emerges when low-frequency partners are included points to low--high coupling, whereas peaks that remain unchanged under such swaps are consistent with predominantly intra-band fundamentals. 
In the mid-frequency region, strict isolation of Amide~II (1450--1600~cm$^{-1}$) and Amide~I (1600--1700~cm$^{-1}$) likewise suppresses cross-band features that appear in broader windows, separating genuinely band-local structure from peaks whose visibility depends on neighbouring-band participation.

\begin{figure}[b]
\centering
\includegraphics[width=\textwidth]{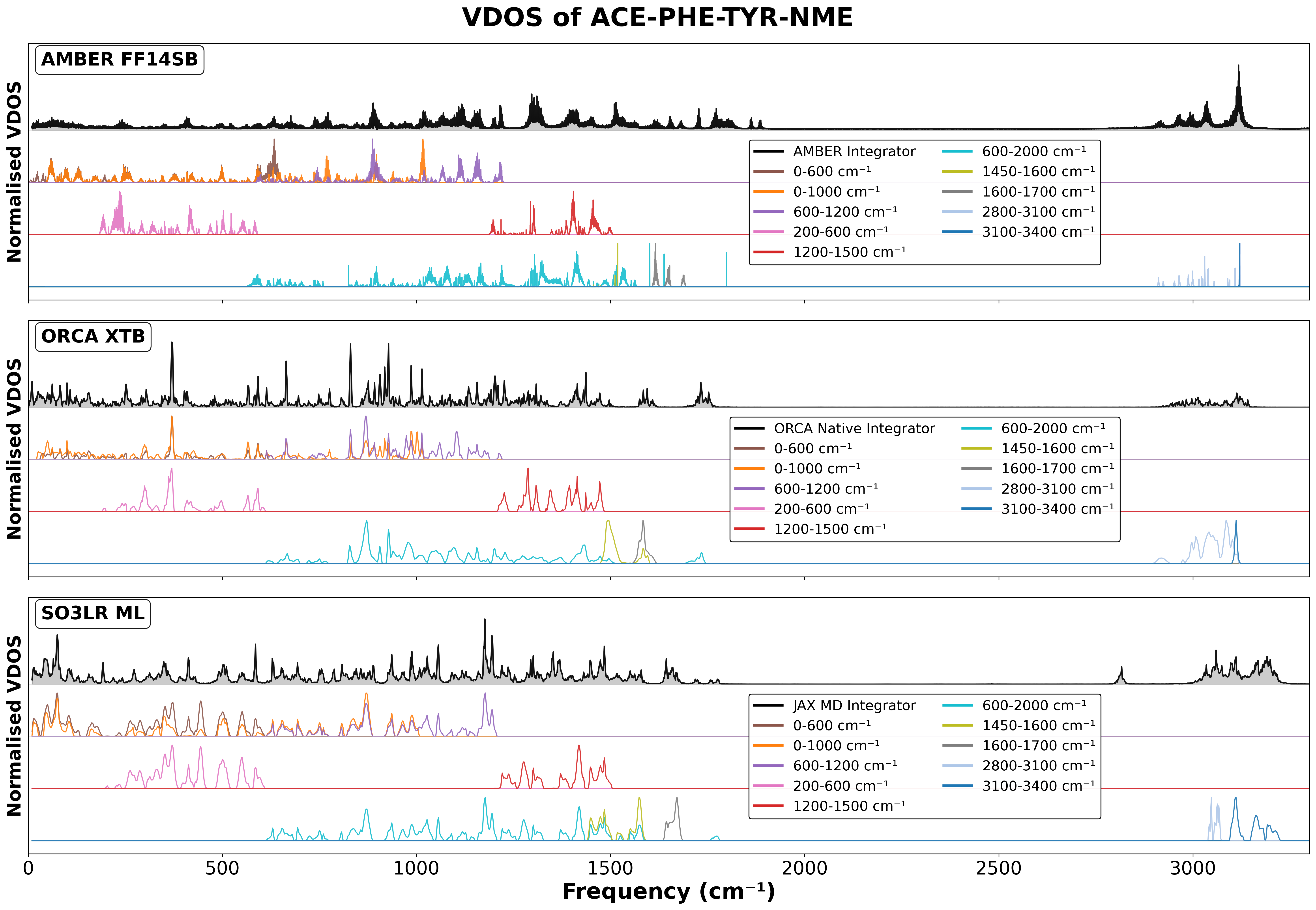}
\caption{\textbf{VDOS of Ace--Phe--Tyr--NMe across force models and band-limited FIMD diagnostics.}
Top to bottom: AMBER ff14SB, ORCA xTB, and SO3LR (JAX MD). Black curves show the reference VDOS from the native Cartesian trajectory (velocity--Verlet; identical $\Delta t$ and total time for all panels). Coloured curves show independent band-limited FIMD propagations restricted to the frequency windows listed in the legend (normalised and vertically offset).}
\label{fig:vdos_three_panel}
\end{figure}
\providecommand{\Lopr}{\hat{\mathcal L}_r}
\providecommand{\Lopp}{\hat{\mathcal L}_p}

\section*{Fourier--expanded Liouville generators and the harmonic phase advance.}
This section records the Fourier--Liouville identities that underlie the band-limited propagators cited in the main text.

\textit{Fourier representation of a degree of freedom.}
For each Cartesian degree of freedom $\ell$, a truncated Fourier representation is introduced,
\begin{equation}
  r^{(\ell)}(t)
  = \sum_{k} r^{(\ell)}_{k}\,e^{i\omega^{(\ell)}_{k} t},
  \label{eq:fourier_r}
\end{equation}
where the index set for $k$ specifies the retained frequencies.
Band limitation corresponds to restricting this index set to frequencies within a chosen window.

\textit{Forcing in Fourier amplitudes.}
Let $L$ denote the Lagrangian associated with the dynamics. The functional derivative with respect to
$r^{(\ell)}(t)$ admits the Fourier-amplitude expansion
\begin{equation}
  \frac{\delta L}{\delta r^{(\ell)}(t)}
  = \sum_{k}
    e^{-i\omega^{(\ell)}_{k} t}\,
    \frac{\partial L}{\partial r^{(\ell)}_{k}},
  \label{eq:deltaL_fourier}
\end{equation}
and the corresponding momentum equation is
\begin{equation}
  \dot p^{(\ell)}(t)
  = \sum_{k}
    e^{-i\omega^{(\ell)}_{k} t}\,
    \frac{\partial L}{\partial r^{(\ell)}_{k}}.
  \label{eq:pdot_fourier}
\end{equation}

\textit{Fourier-expanded split generators.}
Consider the one-degree-of-freedom split generators (associated with the $\hat{\mathcal L}_r$ and
$\hat{\mathcal L}_p$ in the main text) acting on an observable $f$.
Using \eqref{eq:fourier_r} and the chain rule through the Fourier amplitudes gives the positional generator
\begin{align}
  \Lopr^{(\ell)} f
  &=
  \sum_{k,j}
  e^{i(\omega^{(\ell)}_{k}-\omega^{(\ell)}_{j})t}\,
  \Bigl(\dot r^{(\ell)}_{k}
        + i\omega^{(\ell)}_{k} r^{(\ell)}_{k}\Bigr)\,
  \frac{\partial f}{\partial r^{(\ell)}_{j}},
  \label{eq:Lr_fourier}
\end{align}
and inserting \eqref{eq:pdot_fourier} yields the momentum generator
\begin{align}
  \Lopp^{(\ell)} f
  &=
  \sum_{k,j}
  e^{-i\omega^{(\ell)}_{k} t}\,
  \frac{\partial r^{(\ell)}_{j}}{\partial p^{(\ell)}_{k}}\,
  \frac{\partial L}{\partial r^{(\ell)}_{k}}\,
  \frac{\partial f}{\partial r^{(\ell)}_{j}}.
  \label{eq:Lp_fourier}
\end{align}
Equations \eqref{eq:Lr_fourier}--\eqref{eq:Lp_fourier} show explicitly that spectral truncation (restriction
of the Fourier indices $k$) induces a hard band limitation at the operator level.

\textit{Short-time propagators.}
Introduce the compact differential operator
\begin{equation}
\mathcal D^{(\ell,\Delta t)}_{j,m}
:=
\frac{(\Delta t)^m}{m!}\,
\frac{\partial^m}{\partial (r^{(\ell)}_{j})^m}.
\label{eq:Djm_def}
\end{equation}
A Taylor expansion of the exponential maps generated by \eqref{eq:Lr_fourier}--\eqref{eq:Lp_fourier} then gives
\begin{align}
  e^{\Delta t\,\Lopr^{(\ell)}} f
  &=
  \sum_{m=0}^{\infty}\sum_{k,j}
  e^{\,im(\omega^{(\ell)}_{k}-\omega^{(\ell)}_{j})\Delta t}\,
  \mathcal{K}_r^{mk,\ell}\,
  \mathcal D^{(\ell,\Delta t)}_{j,m} f
  \label{eq:expLr_fourier}\\[2pt]
  e^{(\Delta t/2)\,\Lopp^{(\ell)}} f
  &=
  \sum_{m=0}^{\infty}\sum_{k,j}
  e^{-\,\tfrac{i}{2}m\omega^{(\ell)}_{k}\Delta t}\,
  \mathcal{K}^{mk,\ell}_p
  \mathcal D^{(\ell,\Delta t)}_{j,m} f.
  \label{eq:expLp_fourier}
\end{align}
with
\begin{align}
\mathcal{K}_r^{mk,\ell}&=\Bigl(\dot r^{(\ell)}_{k}+ i\omega^{(\ell)}_{k} r^{(\ell)}_{k}\Bigr)^{m}, \\
    \mathcal{K}_p^{mk,\ell}&=\left( \frac{1}{2}\frac{\partial r^{(\ell)}_{j}}{\partial p^{(\ell)}_{k}}\,\frac{\partial L}{\partial r^{(\ell)}_{k}}\right)^{m}. \nonumber
\end{align}
In band-limited dynamics, the sums over $k$ are restricted to the retained Fourier indices; this restriction
is the algebraic expression of hard spectral truncation.

\textit{Harmonic phase advance and recovery of cosine--sine updates.}
In the harmonic limit, each normal-mode coordinate contains only a two-sided single-frequency pair.
For one mode $\nu$ with frequency $\omega_\nu$, write
\begin{equation}
q_\nu(t)=a_\nu e^{i\omega_\nu t}+a_\nu^{*}e^{-i\omega_\nu t}.
\label{eq:q_two_sided}
\end{equation}
A time shift by $\Delta t$ multiplies the Fourier components by $e^{\pm i\omega_\nu\Delta t}$, implying
\begin{equation}
q_\nu(t+\Delta t)
=
q_\nu(t)\cos(\omega_\nu\Delta t)
+\frac{\dot q_\nu(t)}{\omega_\nu}\sin(\omega_\nu\Delta t),
\label{eq:q_shift_cos_sin}
\end{equation}
and, by differentiation,
\begin{equation}
\dot q_\nu(t+\Delta t)
=
\dot q_\nu(t)\cos(\omega_\nu\Delta t)
-\omega_\nu q_\nu(t)\sin(\omega_\nu\Delta t).
\label{eq:qdot_shift_cos_sin}
\end{equation}
Upon identifying $\pi_\nu=\dot q_\nu$ in mass-weighted modal variables, these relations recover the exact
harmonic phase-advance map used for the drift step in the main text: the harmonic generator is diagonal in a
Fourier/modal basis, and its exponential is a pure phase multiplication.

\textit{Band-subspace thermostat.}
If finite-temperature sampling is employed, a thermostat may be applied only to the active band variables.
For example, an Ornstein--Uhlenbeck update for each band momentum $\pi_\nu$ ($\nu\in B$) with friction $\gamma$ is
\begin{equation}
\pi_\nu \leftarrow e^{-\gamma\Delta t}\pi_\nu
+\sqrt{(1-e^{-2\gamma\Delta t})\,k_B T}\;\xi_\nu,
\qquad
\xi_\nu\sim\mathcal N(0,1),
\label{eq:OU_band}
\end{equation}
while $\nu\notin B$ remains unchanged.


\section*{Symplecticity of the band-limited Fourier integrator}

\begin{theorem}[Symplecticity of fixed-reference band-limited FIMD]
\label{thm:supp_symplecticity}
Assume the reference $(\mathbf r_0,H_0,W,\Omega)$ and active set $B$ are held fixed over a step.
Let $P_B$ be the projector onto the retained mode indices, and define the band reconstruction
\begin{equation}
  \mathbf r_B(\mathbf q_B)=\mathbf r_0 + M^{-1/2}W_B\,\mathbf q_B,
  \label{eq:supp_reconstruct_rB}
\end{equation}
with $W_B$ the restriction of $W$ to columns in $B$.
Define the residual potential and residual force by
\begin{equation}
  \Delta V(\mathbf r)=V(\mathbf r)-V(\mathbf r_0)-\tfrac12(\mathbf r-\mathbf r_0)^{\mathsf T}H_0(\mathbf r-\mathbf r_0),
  \qquad
  \mathbf F_\Delta(\mathbf r)=-\nabla\Delta V(\mathbf r),
  \label{eq:supp_DeltaV_def}
\end{equation}
and define the (projected) modal source
\begin{equation}
  \mathbf s_B(\mathbf r):=W_B^{\mathsf T}M^{-1/2}\,\mathbf F_\Delta(\mathbf r).
  \label{eq:supp_source_def}
\end{equation}
Let $\Phi_{0,B}^{\Delta t}$ denote the exact time-$\Delta t$ harmonic flow generated by
$H_{0,B}=\tfrac12\sum_{\nu\in B}(\pi_\nu^2+\omega_\nu^2 q_\nu^2)$.
Consider the kick--drift--kick map
\begin{equation}
\Psi_B^{\Delta t}
:=K_{\Delta t/2}\circ\Phi_{0,B}^{\Delta t}\circ K_{\Delta t/2},
\qquad
K_{\tau}:(\mathbf q_B,\boldsymbol\pi_B)\mapsto\bigl(\mathbf q_B,\boldsymbol\pi_B+\tau\,\mathbf s_B(\mathbf r_B(\mathbf q_B))\bigr).
\label{eq:supp_kdk_map}
\end{equation}
Then $\Psi_B^{\Delta t}$ is symplectic on the band phase space $(\mathbf q_B,\boldsymbol\pi_B)$ with canonical form
$\omega_B=\sum_{\nu\in B} dq_\nu\wedge d\pi_\nu$.
\end{theorem}

\begin{proof}
Define the band Hamiltonian
\[
H_B(\mathbf q_B,\boldsymbol\pi_B)=H_{0,B}(\mathbf q_B,\boldsymbol\pi_B)+H_{\Delta,B}(\mathbf q_B),
\qquad
H_{\Delta,B}(\mathbf q_B):=\Delta V\bigl(\mathbf r_B(\mathbf q_B)\bigr).
\]

(i) $\Phi_{0,B}^{\Delta t}$ is the exact Hamiltonian flow of $H_{0,B}$, hence symplectic.

(ii) For $H_{\Delta,B}$, Hamilton's equations give $\dot{\mathbf q}_B=\nabla_{\boldsymbol\pi_B}H_{\Delta,B}=0$ and
$\dot{\boldsymbol\pi}_B=-\nabla_{\mathbf q_B}H_{\Delta,B}$.
Using \eqref{eq:supp_reconstruct_rB} and the chain rule,
\[
\nabla_{\mathbf q_B}H_{\Delta,B}(\mathbf q_B)
=(M^{-1/2}W_B)^{\mathsf T}\nabla_{\mathbf r}\Delta V(\mathbf r_B)
=-(W_B^{\mathsf T}M^{-1/2})\,\mathbf F_\Delta(\mathbf r_B)
=-\mathbf s_B(\mathbf r_B),
\]
so $\dot{\boldsymbol\pi}_B=\mathbf s_B(\mathbf r_B(\mathbf q_B))$.
Therefore the exact time-$\tau$ flow of $H_{\Delta,B}$ is precisely the shear $K_{\tau}$ in \eqref{eq:supp_kdk_map}, and hence is symplectic.

(iii) The full step $\Psi_B^{\Delta t}$ is the Strang composition of symplectic maps in \eqref{eq:supp_kdk_map}, so it is symplectic.
\end{proof}

\begin{figure}[hp]
  \centering
  \includegraphics[width=0.55\columnwidth]{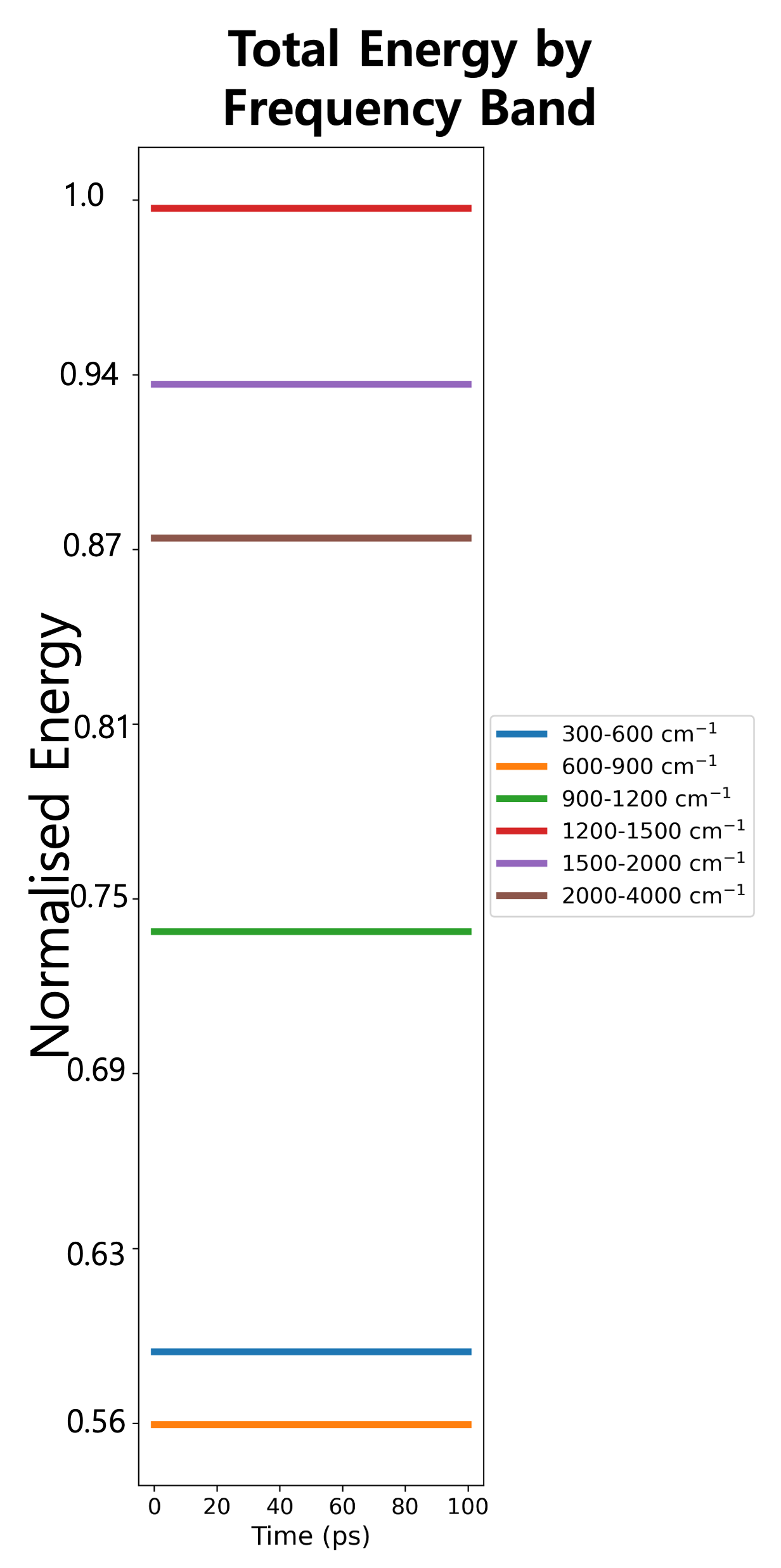}
  \caption{%
    Normalised (where the maximum energy is set to 1) total band energies
    \(E_B(t)\) for band-limited FIMD NVE trajectories of
    Ace--Phe--Tyr--NMe over \(100~\mathrm{ps}\).
    Each horizontal trace corresponds to one of the six frequency windows
    indicated in the legend.
    The near-perfect flatness of all curves indicates that the band-limited
    dynamics conserve total energy within each retained frequency band to
    numerical precision.}
  \label{fig:total_energy_by_band}
\end{figure}

\begin{figure*}[htbp]
  \centering
  \includegraphics[width=0.9\textwidth]{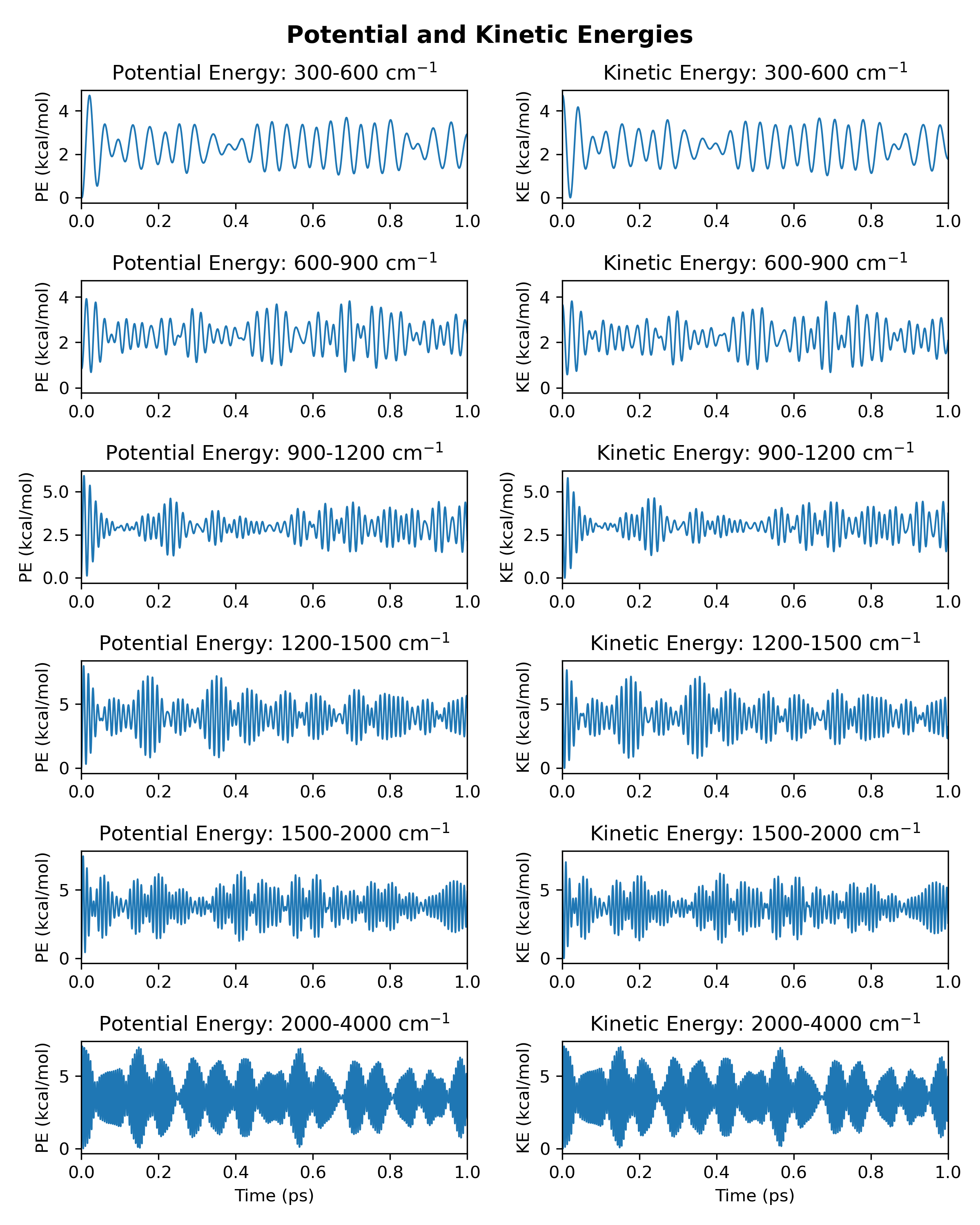}
  \caption{%
    Band-resolved potential (left column) and kinetic (right column) energies
    for the same band-limited NVE trajectories as in
    Fig.~\ref{fig:total_energy_by_band}, over the first
    \(1~\mathrm{ps}\).
    Each row corresponds to one frequency window
    \(B\in\{300\text{--}600, 600\text{--}900, 900\text{--}1200,
    1200\text{--}1500, 1500\text{--}2000, 2000\text{--}4000~\mathrm{cm}^{-1}\}\).
    The traces show near-harmonic exchange between potential and kinetic
    components with a band-dependent modulation pattern, while the total
    band energy \(E_B(t)\) remains essentially constant as in
    Fig.~\ref{fig:total_energy_by_band}.
  }
  \label{fig:pe_ke_by_band}
\end{figure*}

\section{Remark on NVE band-limited tests and sonification of the potential energy profiles}

Figure~\ref{fig:total_energy_by_band} summarises the band-resolved NVE diagnostics
for Ace--Phe--Tyr--NMe under band-limited FIMD.
For each frequency window
\begin{equation}
    B\in\{300\text{--}600~\mathrm{cm}^{-1}, 600\text{--}900~\mathrm{cm}^{-1}, 900\text{--}1200~\mathrm{cm}^{-1}, 1200\text{--}1500~\mathrm{cm}^{-1}, 1500\text{--}2000~\mathrm{cm}^{-1}, 2000\text{--}4000~\mathrm{cm}^{-1}\},
\end{equation}
the total band energy \(E_B(t)=K_B(t)+V_B(t)\) remains essentially constant over \(100~\mathrm{ps}\), with no resolvable secular drift at the plotting scale.
This is consistent with the interpretation of the band-limited Liouville generator as a microcanonical evolution on the selected vibrational subspace:
the exact harmonic block (the FIMD rotation step) does not introduce numerical heating, and residual energy exchange with discarded modes is negligible on these timescales.

The short-time (1 ps) structure of intra-band energy exchange is shown in Fig.~\ref{fig:pe_ke_by_band}.
For each window, the band-resolved potential and kinetic energies,
\(V_B(t)\) and \(K_B(t)\), display strongly correlated oscillations with a nearly constant sum \(E_B(t)\).
The modulation pattern becomes progressively faster and more structured with increasing frequency, but the character remains that of a near-harmonic exchange between \(V_B\) and \(K_B\) on the invariant circles generated by the harmonic Liouville block.
This behaviour is what one expects from the exact rotation step in the \((q_\nu,p_\nu)\)-planes combined with weak anharmonic source terms.

To provide a perceptual representation of this band-resolved exchange, the time series \(V_B(t)\) underlying
Fig.~\ref{fig:pe_ke_by_band} were mapped to audio waveforms by linear time rescaling and amplitude normalisation.
Representative audio files are available at
\begin{equation}
  \texttt{\url{https://kyunghoon-han.github.io/FIMD_sonification.html}},
\end{equation}
allowing the reader to listen to the band-limited NVE dynamics generated by the FIMD scheme.
\subsection{NVT stationarity via Maxwell--Boltzmann velocity statistics (Ace--Phe--Tyr--NMe)}

A basic canonical (NVT) validation is that the atomic velocities sample the Maxwell--Boltzmann (MB) distribution at the target
temperature. For a particle of mass $m$ in three dimensions, the MB \emph{speed} distribution is
\begin{equation}
p(v; m,T)=4\pi \left(\frac{m}{2\pi k_{\mathrm B}T}\right)^{3/2} v^{2}\exp\!\left(-\frac{mv^{2}}{2k_{\mathrm B}T}\right),
\label{eq:MB_speed}
\end{equation}
and for a molecular system with multiple atomic masses the all-atom speed distribution is a mixture
$p_{\mathrm{mix}}(v)=\sum_i w_i\,p(v;m_i,T)$ with weights $w_i$ proportional to the counts of each atomic species.

\begin{figure}[htb]
  \centering
  \includegraphics[width=\textwidth]{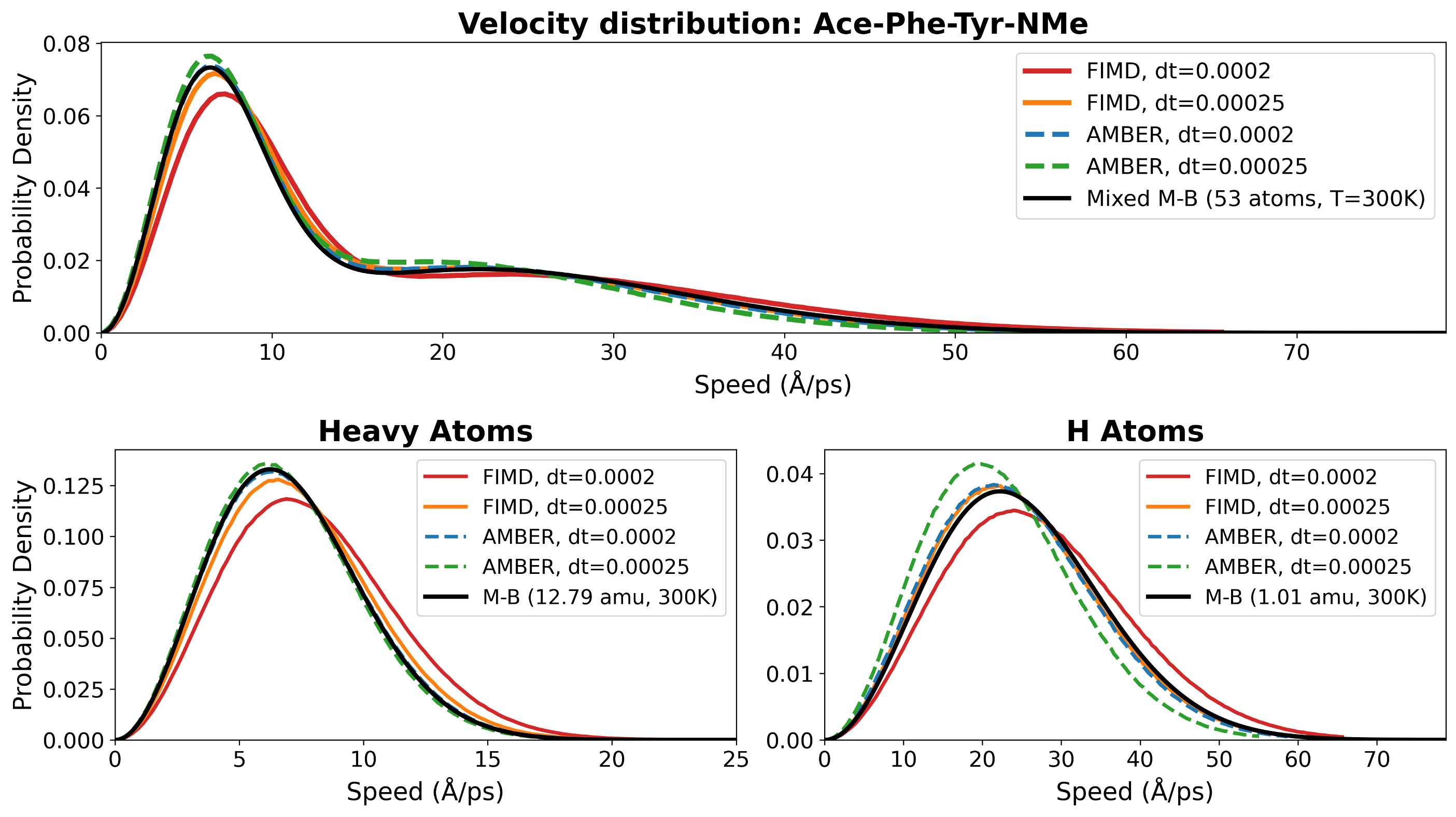}
  \caption{\textbf{Canonical (NVT) velocity statistics for Ace--Phe--Tyr--NMe at $T=300$ K.}
  Top: all-atom speed distributions from FIMD and a conventional AMBER trajectory, each run at two time steps
  ($\Delta t=0.0002$ and $0.00025$ ps), compared to the corresponding analytic mixed Maxwell--Boltzmann distribution
  (black; 53 atoms).
  Bottom: mass-resolved speed distributions for heavy atoms (left; effective mass $12.79$ amu) and hydrogen atoms (right;
  mass $1.01$ amu), compared to their analytic MB curves (black). Agreement of both the mixture and mass-resolved marginals
  indicates stationary sampling of the canonical momentum distribution at the target temperature.}
  \label{fig:velocity_MB_test}
\end{figure}

Figure~\ref{fig:velocity_MB_test} shows that FIMD reproduces the expected MB speed statistics in NVT, both for the full
all-atom mixture and when separating heavy atoms from hydrogens (which have distinctly different most-probable speeds due
to the mass dependence in Eq.~\eqref{eq:MB_speed}). The FIMD curves closely track the conventional MD reference and the
corresponding analytic MB targets over the full range of speeds; any visible differences are confined to the far tails,
where finite-sample noise is largest. This supports the canonical stationarity of the band-subspace thermostat in practice,
as assessed through Cartesian velocity marginals.

\begin{figure}[tb]
\centering
\includegraphics[width=\textwidth]{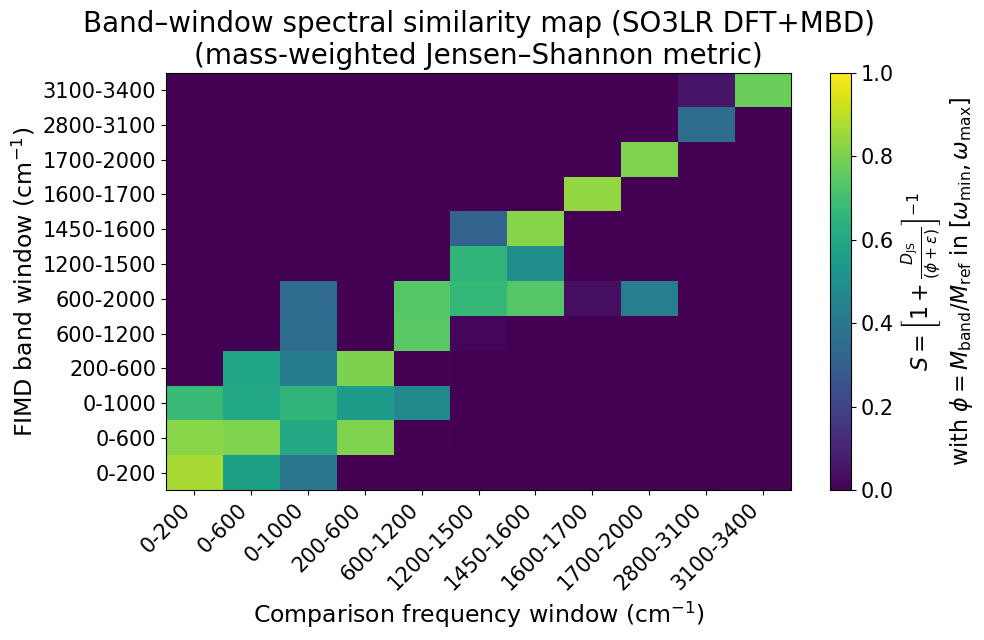}
\caption{\textbf{Frequency-windowed Jensen--Shannon comparison between band-limited FIMD and conventional SO3LR VDOS.} Rows denote band-limited FIMD runs and columns denote comparison windows $[\omega_{\min},\omega_{\max}]$. Within each window we normalise the spectral mass, compute the Jensen--Shannon distance, and report the mass-weighted similarity score $S$ (Eq.~\ref{eq:score}). Correspondence is highest along the diagonal, while off-diagonal agreement appears when the comparison window overlaps the propagated band. Quantitatively, the diagonal entries dominate: $\langle S\rangle_{\mathrm{diag}}=0.74$ (var.\ $2.1\times10^{-2}$) versus $\langle S\rangle_{\mathrm{off}}=0.089$ (var.\ $4.5\times10^{-2}$), giving a diagonal/off-diagonal mean ratio of $\sim$8.3. This diagonal dominance is the strongest among the SO3LR, ORCA xTB, and AMBER ff14SB comparisons.}
\label{fig:so3lr_comparison}
\end{figure}

\begin{figure}[tb]
\centering
\includegraphics[width=\textwidth]{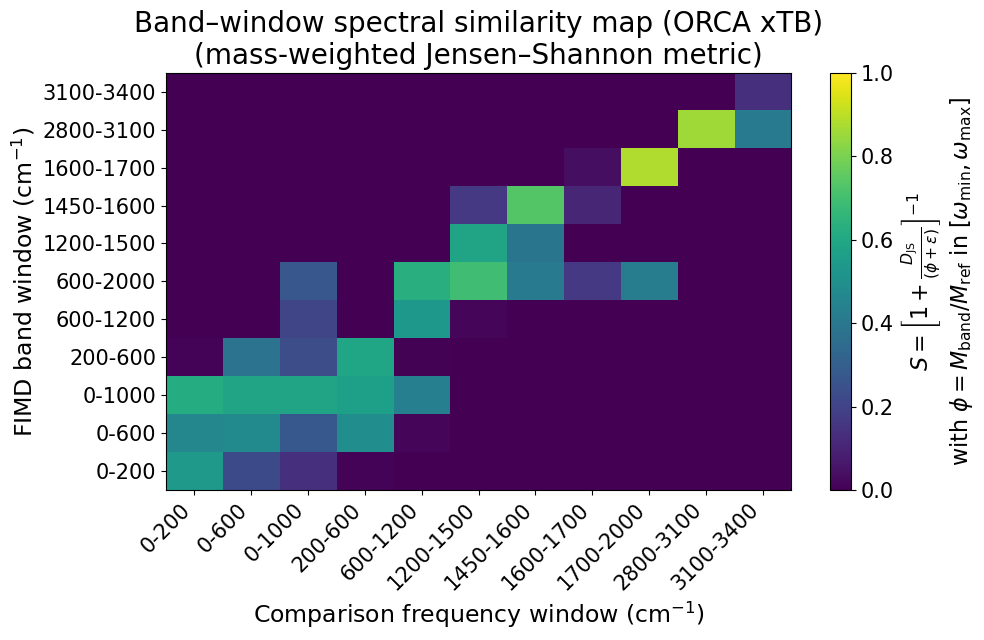}
\caption{\textbf{Frequency-windowed Jensen--Shannon comparison between band-limited FIMD and conventional ORCA (xTB) VDOS.} Rows denote band-limited FIMD runs and columns denote comparison windows $[\omega_{\min},\omega_{\max}]$. Within each window we normalise the spectral mass, compute the Jensen--Shannon distance, and report the mass-weighted similarity score $S$ (Eq.~\ref{eq:score}). Correspondence is strongest along the diagonal, with off-diagonal agreement when the comparison window overlaps the propagated band; compared with SO3LR, the low-frequency windows show systematically lower scores. Quantitatively, the diagonal remains dominant: $\langle S\rangle_{\mathrm{diag}}=0.51$ (var.\ $6.2\times10^{-2}$) versus $\langle S\rangle_{\mathrm{off}}=0.083$ (var.\ $3.4\times10^{-2}$), giving a diagonal/off-diagonal mean ratio of $\sim$6.1. This diagonal tendency is weaker than for SO3LR but stronger than for AMBER ff14SB.}
\label{fig:orca_comparison}
\end{figure}

\begin{figure}[tb]
\centering
\includegraphics[width=\textwidth]{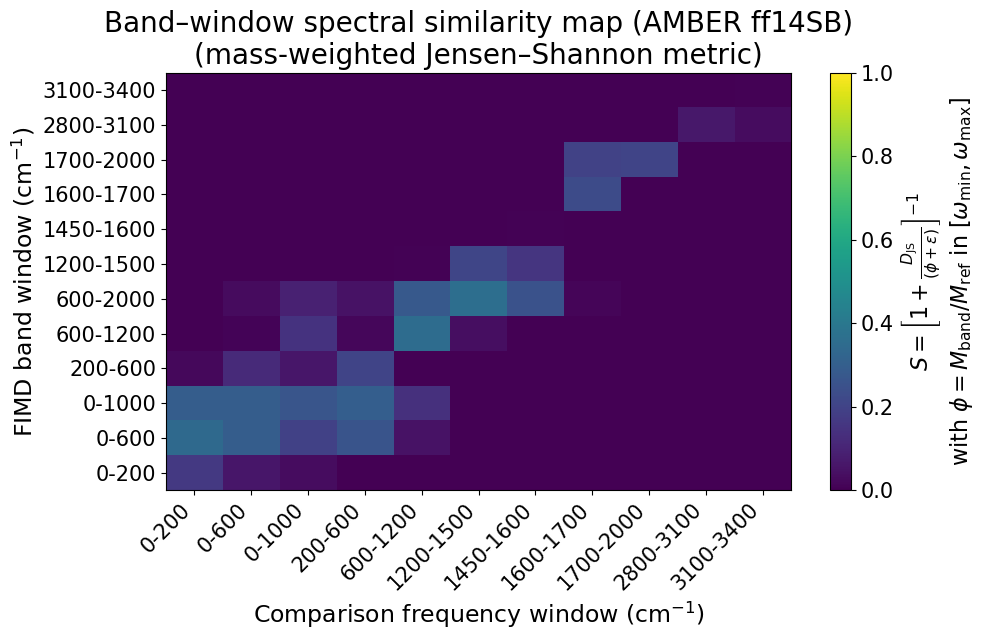}
\caption{\textbf{Frequency-windowed Jensen--Shannon comparison between band-limited FIMD and conventional AMBER VDOS.} Rows denote band-limited FIMD runs and columns denote comparison windows $[\omega_{\min},\omega_{\max}]$. Within each window we normalise the spectral mass, compute the Jensen--Shannon distance, and report the mass-weighted similarity score $S$ (Eq.~\ref{eq:score}). Correspondence is strongest along the diagonal, with off-diagonal agreement when the comparison window overlaps the propagated band; overall scores are lower than for ORCA (xTB) and SO3LR. Quantitatively, the diagonal remains dominant but with reduced absolute agreement: $\langle S\rangle_{\mathrm{diag}}=0.18$ (var.\ $1.3\times10^{-2}$) versus $\langle S\rangle_{\mathrm{off}}=0.032$ (var.\ $6.5\times10^{-3}$), giving a diagonal/off-diagonal mean ratio of $\sim$5.7. Note that, although Fig.~\ref{fig:vdos_three_panel} shows clear peak alignment in the 3100--3400~cm$^{-1}$ region, the corresponding FIMD peaks are markedly narrower than in the conventional trajectory, which strongly reduces the windowed similarity score despite the apparent agreement in peak positions.}
\label{fig:amber_comparison}
\end{figure}

\section*{VDOS similarity between band-limited FIMD runs vs. conventional MD runs per band}

Figures \ref{fig:so3lr_comparison}, \ref{fig:orca_comparison}, and \ref{fig:amber_comparison} show the band-wise similarity between the VDOS from band-limited FIMD (y-axis) and the VDOS from the corresponding conventional velocity–Verlet trajectories (x-axis) for each force field. 
Similarity is quantified via the Jensen–Shannon (JS) distance evaluated over selected comparison windows (x-axis). For each window $[\omega_{\min},\omega_{\max}]$, we restrict both spectra to that interval (discarding all contributions outside it), interpolate them onto a common frequency grid, and convert the windowed spectra into probability distributions by normalising the integrated spectral mass to unity. Concretely, if $s(\omega)$ denotes the windowed VDOS sampled on grid points ${\omega_i}$ with spacings $\Delta\omega_i$, we form the discrete probability mass
\begin{equation}
p_i=\frac{s(\omega_i),\Delta\omega_i}{\sum_j s(\omega_j),\Delta\omega_j},
\qquad \omega_i\in[\omega_{\min},\omega_{\max}],
\end{equation}
and analogously $q_i$ for the band-limited FIMD spectrum. The JS distance is then computed as
\begin{equation}
D_{\mathrm{JS}}(p,q)=\sqrt{\tfrac12 D_{\mathrm{KL}}(p,|,m)+\tfrac12 D_{\mathrm{KL}}(q,|,m)},
\qquad m=\tfrac12(p+q),
\end{equation}
with $D_{\mathrm{KL}}(p,|,m)=\sum_i p_i\log_2\big(p_i/m_i\big)$.

To incorporate not only shape agreement but also whether a given band-limited run places appreciable spectral weight in the comparison window, we additionally report a mass-weighted similarity score. Define the window mass for spectrum $s$ as
\begin{equation}
M_{[\omega_{\min},\omega_{\max}]}(s)=\int_{\omega_{\min}}^{\omega_{\max}} s(\omega),d\omega,
\end{equation}
and the corresponding mass fraction relative to the conventional reference spectrum $s_{\mathrm{ref}}$ as
\begin{equation}
\phi=\frac{M_{[\omega_{\min},\omega_{\max}]}(s_{\mathrm{band}})}{M_{[\omega_{\min},\omega_{\max}]}(s_{\mathrm{ref}})}.
\end{equation}
We then define the combined similarity
\begin{equation} \label{eq:score}
S=\frac{1}{1 + D_{\mathrm{JS}}(p,q)/(\phi+\varepsilon)},
\end{equation}
where $\varepsilon>0$ prevents division by zero when the window mass is negligible. By construction, $S\in(0,1]$ increases with improving spectral-shape agreement (smaller $D_{\mathrm{JS}}$) and decreases when the band-limited spectrum carries little weight in the window.
We choose this score because the JS distance alone can be deceptively small when both windowed spectra are dominated by weak, noise-like tails after normalisation; incorporating the mass fraction, $\phi$, down-weights such cases and highlights windows where the band-limited dynamics both reproduces the reference spectral shape and captures a non-negligible portion of the reference spectral weight.

Across all three force fields, the band--window maps show a clear diagonal: the highest similarity occurs when the comparison window coincides with (or lies within) the propagated FIMD band, as expected for band-limited dynamics. Off-diagonal agreement is also present when the comparison window overlaps the propagated band.

In the low-frequency region, the scores are systematically lowest for AMBER ff14SB, intermediate for ORCA xTB, and highest for SO3LR (DFT+MBD).
This ordering is consistent with the qualitative structure of the corresponding low-band dynamics in Fig.~2(d)--(f): AMBER ff14SB exhibits the most space-filling, weakly organised projected trajectories, ORCA xTB shows intermediate coherence, and SO3LR displays the most regular, Lissajous-like phase-space loops.
Taken together, the JS-based window scores and the 2D PCA orbits indicate a progressive sharpening of well-defined low-frequency collective motion from ff14SB to xTB to SO3LR, i.e.\ increasingly coherent oscillatory content and reduced low-frequency “noise-like” spreading in the projected band dynamics.

For AMBER ff14SB, we also note that despite clear peak alignment in the 3100--3400~cm$^{-1}$ region (Fig.~\ref{fig:vdos_three_panel}), the corresponding FIMD peaks are markedly narrower than in the conventional trajectory. This width mismatch strongly reduces the windowed similarity score even when peak positions agree.

\section*{Phase mutual-information maps of Fourier-resolved modes}

\begin{figure}[t]
  \centering
  \includegraphics[width=0.86\linewidth]{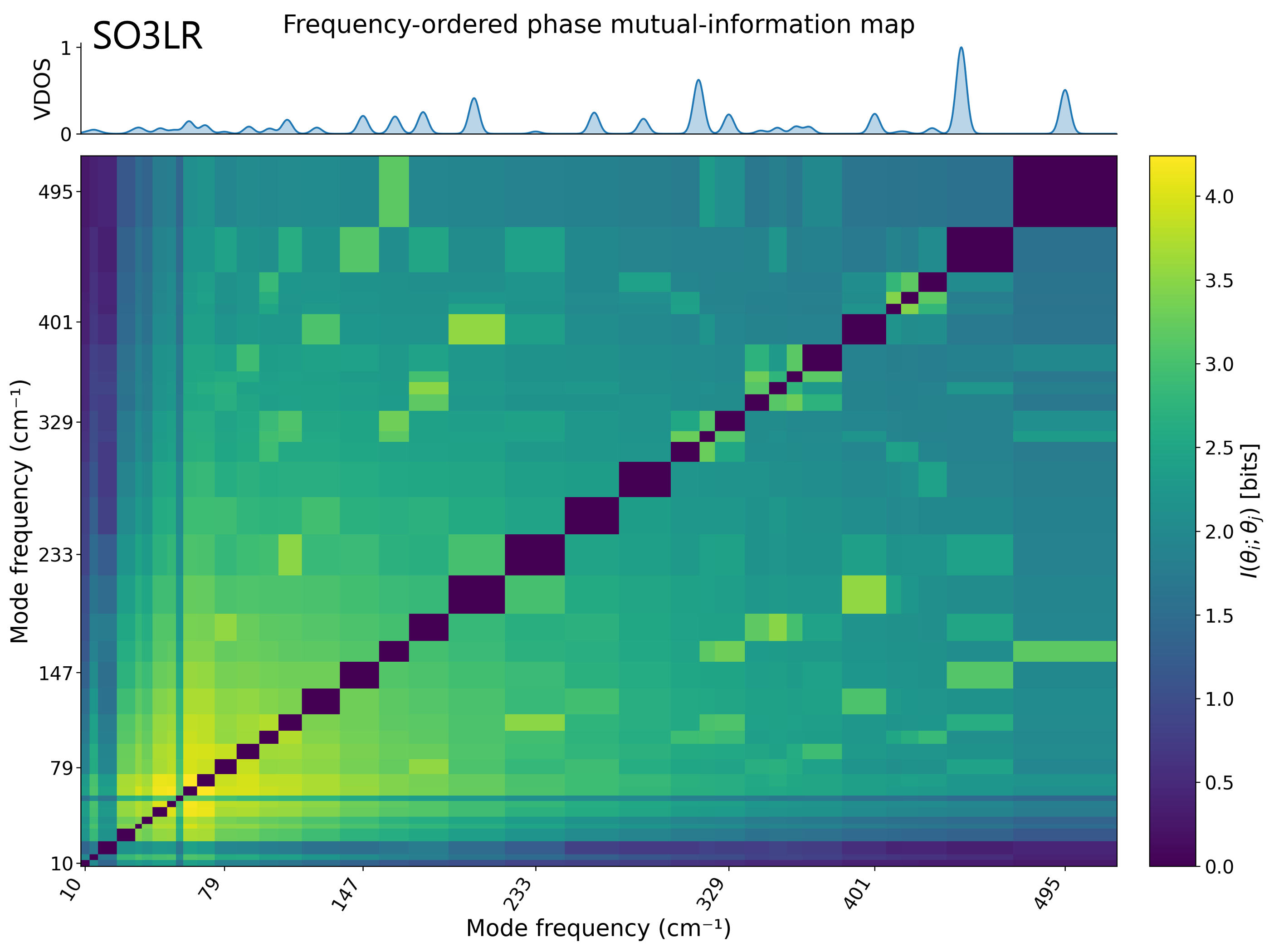}
  \caption{
  Frequency-ordered phase mutual-information map for Ace--Phe--Tyr--NMe
  propagated with SO3LR. The map shows a pronounced low-frequency block
  together with extended low-to-mid-frequency phase dependencies. Within the
  present estimator, this indicates that SO3LR resolves the most clearly
  organised soft-mode phase structure among the three force models considered
  here, with low-frequency modes retaining phase information that is shared
  with broader anharmonic motion.
  }
  \label{fig:supp_phase_mi_so3lr}
\end{figure}

\begin{figure}[t]
  \centering
  \includegraphics[width=0.86\linewidth]{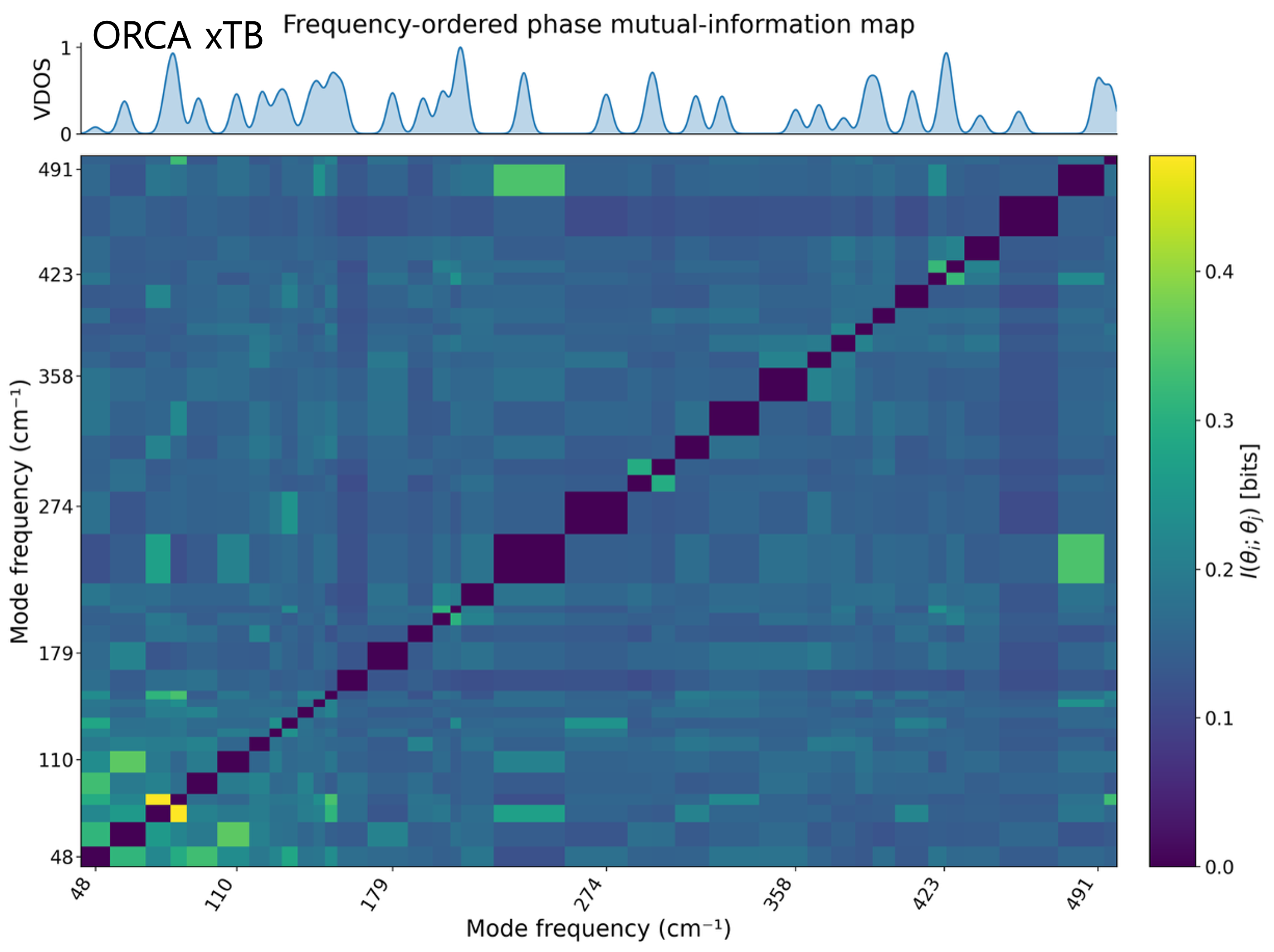}
  \caption{
  Frequency-ordered phase mutual-information map for Ace--Phe--Tyr--NMe
  propagated with ORCA xTB. Compared with SO3LR, the xTB map shows a less
  concentrated low-frequency block and a more diffuse background of finite
  phase mutual information. This corresponds to an intermediate level of
  low-frequency phase organisation: appreciable phase dependence is present
  across the low- and mid-frequency manifold, but the soft-mode structure is
  less sharply localised than in SO3LR.
  }
  \label{fig:supp_phase_mi_xtb}
\end{figure}

\begin{figure}[t]
  \centering
  \includegraphics[width=0.86\linewidth]{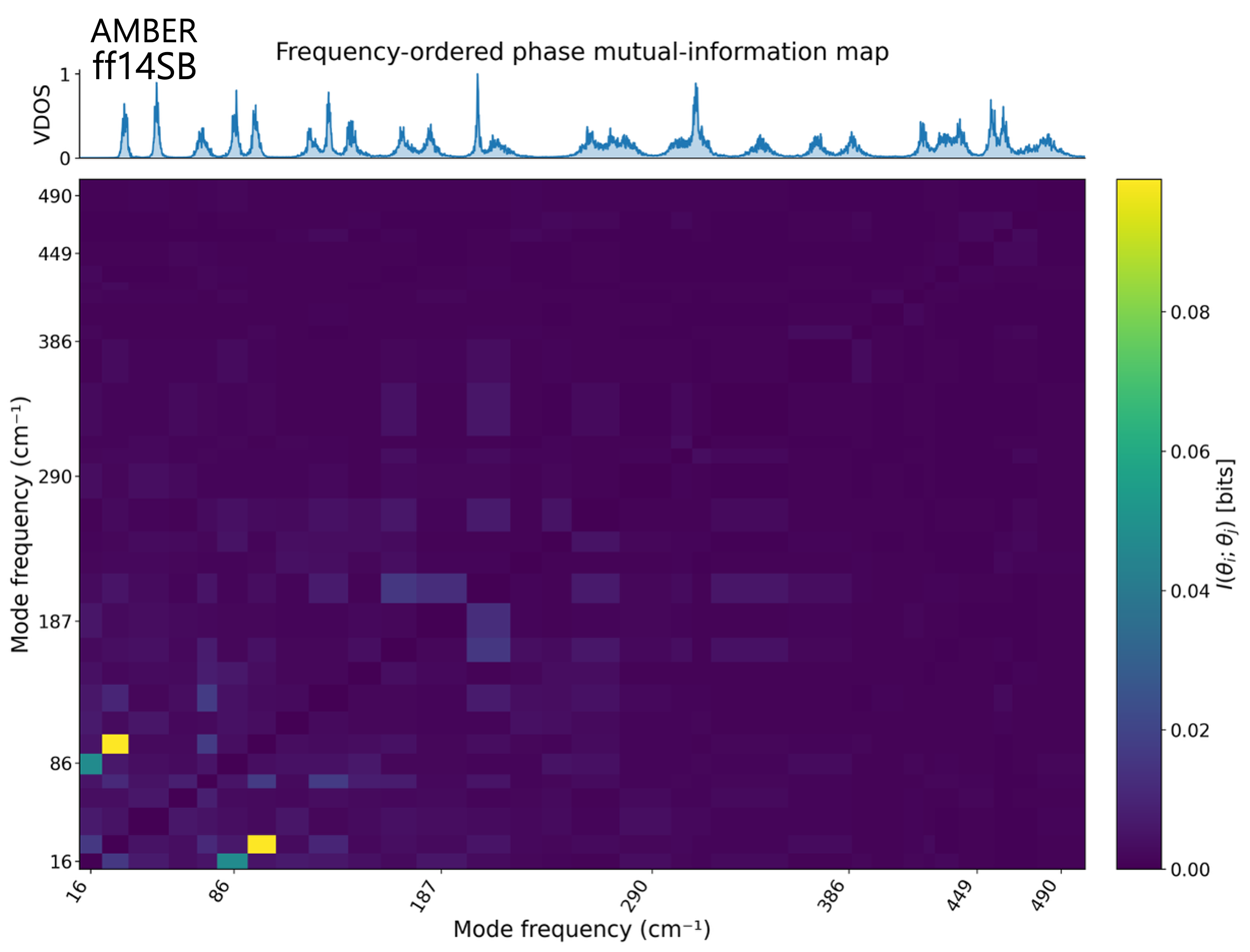}
  \caption{
  Frequency-ordered phase mutual-information map for Ace--Phe--Tyr--NMe
  propagated with AMBER ff14SB. Compared with SO3LR and ORCA xTB, the AMBER
  map is the sparsest, with weak off-diagonal phase dependence concentrated
  mainly in a limited subset of low-frequency modes. In the present comparison,
  this indicates a more weakly organised low-frequency phase-coupling structure
  and dynamics closer to a separable modal picture within this frequency
  range.
  }
  \label{fig:supp_phase_mi_amber}
\end{figure}

Figures~\ref{fig:supp_phase_mi_so3lr}--\ref{fig:supp_phase_mi_amber}
show that the three force models differ not only in their VDOS peak positions
and intensities, but also in the organisation of their nonlinear
phase-coupling patterns. SO3LR gives the most clearly structured
low-frequency manifold, with soft-mode phase organisation extending into
low-to-mid-frequency couplings. ORCA xTB shows intermediate behaviour, with
appreciable but more diffuse phase dependence. AMBER ff14SB yields the
sparsest phase-dependence map, indicating weaker and more localised
off-diagonal coupling in this frequency range. Thus, phase mutual information
provides a complementary force-field comparison metric: beyond reproducing
spectral peaks, it tests how clearly a model resolves the collective
low-frequency organisation and nonlinear coupling structure of the vibrational
dynamics. This comparison should not be read as an absolute ranking of force
fields; rather, it distinguishes how each model represents low-frequency phase
coherence under the same analysis protocol.

The maps were constructed as follows. For each force model
$\alpha\in\{\mathrm{SO3LR},\mathrm{ORCA~xTB},\mathrm{AMBER~ff14SB}\}$,
the trajectory was expressed in mass-weighted modal coordinates after removal
of overall translation and rotation:
\begin{equation}
  x^{(\alpha)}(t_n)
  =
  M^{1/2}\left(r^{(\alpha)}(t_n)-r_0^{(\alpha)}\right),
  \qquad
  q_\nu^{(\alpha)}(t_n)
  =
  \left(w_\nu^{(\alpha)}\right)^T x^{(\alpha)}(t_n),
\end{equation}
with conjugate modal momentum
\begin{equation}
  \pi_\nu^{(\alpha)}(t_n)
  =
  \left(w_\nu^{(\alpha)}\right)^T M^{-1/2}p^{(\alpha)}(t_n).
\end{equation}
The VDOS traces were detrended and decomposed into Lorentzian components using
the TIHI toolkit~\cite{HanTihi2024}; the resulting Lorentzian peak centres
define the frequency markers $\widetilde{\nu}_\nu^{(\alpha)}$ used in the
mode-resolved plots. The corresponding angular frequency is
\begin{equation}
  \omega_\nu^{(\alpha)}
  =
  2\pi c\,\widetilde{\nu}_\nu^{(\alpha)} .
\end{equation}

The instantaneous phase of mode $\nu$ was then defined from the scaled modal
phase-plane coordinate
\begin{equation}
  z_\nu^{(\alpha)}(t_n)
  =
  q_\nu^{(\alpha)}(t_n)
  -
  i\,\frac{\pi_\nu^{(\alpha)}(t_n)}
          {\omega_\nu^{(\alpha)}},
  \qquad
  \theta_\nu^{(\alpha)}(t_n)
  =
  \arg z_\nu^{(\alpha)}(t_n)
  \in [0,2\pi).
\end{equation}
This convention corresponds to the harmonic parametrisation
$q_\nu=A_\nu\cos\theta_\nu$ and
$\pi_\nu=-A_\nu\omega_\nu\sin\theta_\nu$; changing the overall sign convention
does not change the mutual-information estimate.

For each pair of modes $(i,j)$, the circular phase samples
$\{(\theta_i(t_n),\theta_j(t_n))\}_{n=1}^{N_t}$ were binned on
$[0,2\pi)\times[0,2\pi)$, giving empirical probabilities
$\widehat p_{ij}^{ab}$ and marginals $\widehat p_i^a,\widehat p_j^b$.
The plotted matrix entry is the phase mutual information
\begin{equation}
  \widehat I_{ij}^{(\alpha)}
  =
  \sum_{a,b}
  \widehat p_{ij}^{ab}
  \log_2
  \left(
  \frac{\widehat p_{ij}^{ab}}
       {\widehat p_i^a\widehat p_j^b}
  \right),
\end{equation}
reported in bits. Diagonal entries were set to zero, since
$I(\theta_i;\theta_i)$ is a marginal phase entropy rather than an inter-mode
coupling measure.

The heat-map axes are physical frequency axes rather than uniform mode-index
axes. Cell boundaries were placed halfway between neighbouring peak centres,
\begin{equation}
  b_{\nu+1/2}
  =
  \frac{\widetilde{\nu}_\nu+\widetilde{\nu}_{\nu+1}}{2},
\end{equation}
with one-sided extensions at the endpoints. Thus the unequal cell sizes reflect
the non-uniform spacing of the Fourier-resolved vibrational frequencies, not a
plotting artefact.

In a separable harmonic system with factorised phase statistics, the
off-diagonal mutual information vanishes:
\begin{equation}
  p(\theta_i,\theta_j)=p_i(\theta_i)p_j(\theta_j),
  \qquad
  I(\theta_i;\theta_j)=0
  \quad (i\neq j).
\end{equation}
Non-zero off-diagonal entries therefore indicate statistical phase dependence
between modes. In molecular dynamics, such dependence can arise from
anharmonic coupling, phase locking, common driving, or multi-mode collective
organisation. Consequently, the phase mutual-information map should be read as
a nonlinear statistical diagnostic of vibrational phase organisation; a finite
pairwise MI value does not, by itself, imply a direct pairwise force-constant
coupling.
\section{Time-step stability analysis for low-frequency SO3LR FIMD}
\label{sec:supp-timestep-stability}

To assess the practical time-step stability of the band-limited Fourier
integrator, we performed a stability scan for the
\(0\!-\!200~\mathrm{cm}^{-1}\) SO3LR FIMD trajectory of
Ace--Phe--Tyr--NMe.  This low-frequency window was chosen because it is
the regime where band limitation most strongly relaxes the time-scale
separation relative to conventional Cartesian molecular dynamics, while
still retaining thermodynamically important soft collective modes.

Figure~\ref{fig:supp-timestep-stability} shows the relative total-energy
deviation,
\[
  \frac{\Delta E(t)}{E_0}
  =
  \frac{E(t)-E_0}{E_0},
\]
for conventional Cartesian MD and for band-limited FIMD runs using
different time steps.  The conventional MD reference was propagated with
\(\Delta t=1.0~\mathrm{fs}\).  The FIMD trajectories were propagated in
the \(0\!-\!200~\mathrm{cm}^{-1}\) band using
\(\Delta t=0.5\), \(2.5\), \(4.0\), and \(4.5~\mathrm{fs}\).

The \(0.5\), \(2.5\), and \(4.0~\mathrm{fs}\) FIMD trajectories remain
bounded over the simulation window, with no visible secular energy
growth at the scale of the plot.  By contrast, the
\(\Delta t=4.5~\mathrm{fs}\) trajectory displays a rapid instability
near the beginning of the run, visible as a sharp growth of
\(\Delta E/E_0\).  For the present SO3LR peptide system and
implementation, this places the practical stability threshold between
\(4.0\) and \(4.5~\mathrm{fs}\) for the \(0\!-\!200~\mathrm{cm}^{-1}\)
band.

This stability limit should not be interpreted as a simple Nyquist
bound.  Although band limitation removes high-frequency modes from the
propagated subspace, the actual stability is controlled by the
second-order splitting error, the residual anharmonic force, and
nonlinear coupling within the retained band.  The time-step scan
therefore provides an empirical operational criterion for the present
system: low-frequency FIMD can be propagated stably with substantially
larger time steps than the conventional reference, but the admissible
step size remains limited by the residual-force dynamics rather than by
frequency sampling alone.

\begin{figure}[t]
  \centering
  \includegraphics[width=0.92\linewidth]{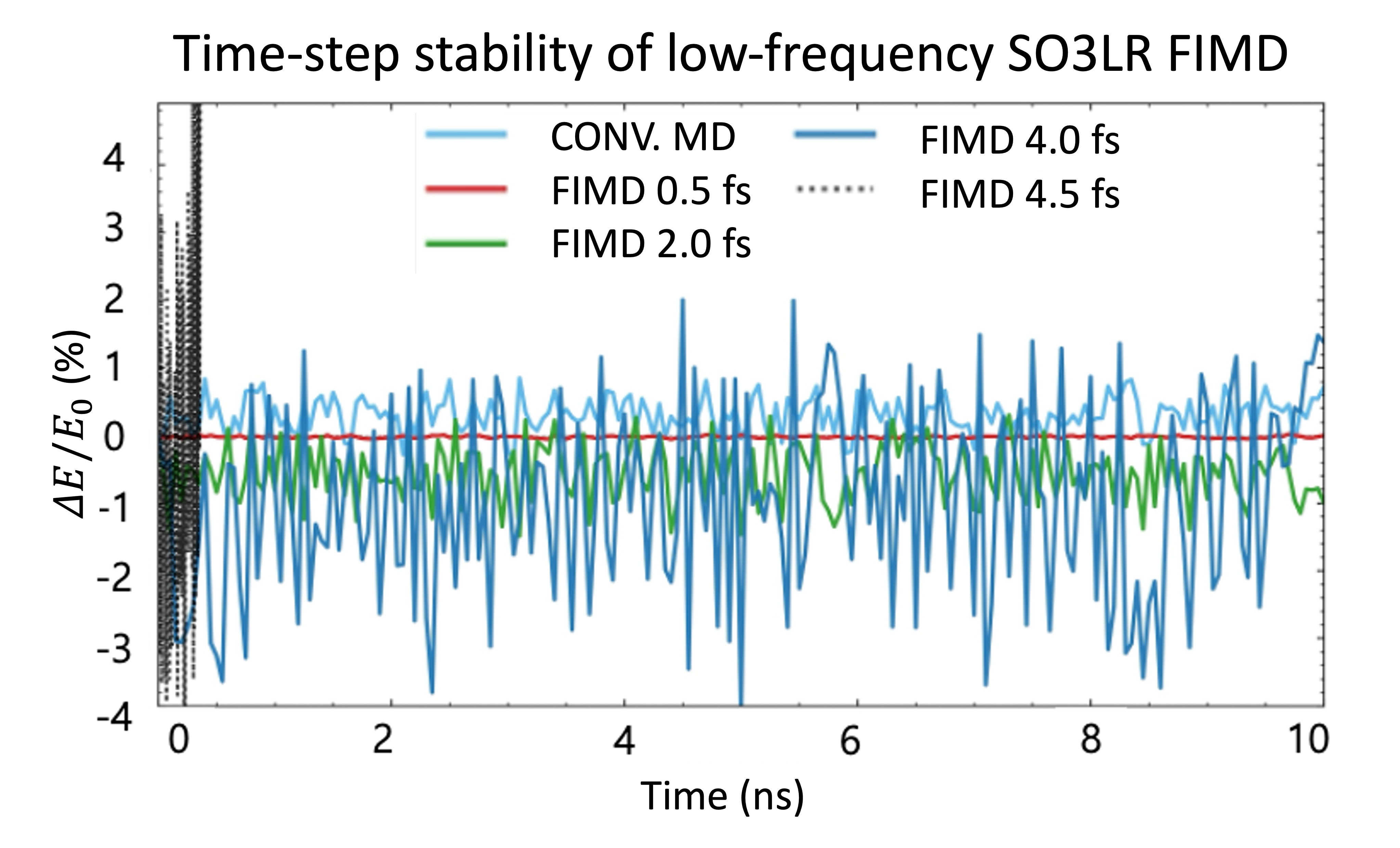}
  \caption{
  \textbf{Time-step stability of low-frequency SO3LR FIMD.}
  Relative total-energy deviation,
  \(\Delta E/E_0\), for conventional Cartesian MD and for
  \(0\!-\!200~\mathrm{cm}^{-1}\) band-limited FIMD trajectories of
  Ace--Phe--Tyr--NMe propagated with different time steps.  The
  conventional MD reference uses \(\Delta t=1.0~\mathrm{fs}\).  The
  FIMD trajectories with \(\Delta t=0.5\), \(2.5\), and
  \(4.0~\mathrm{fs}\) remain bounded over the plotted interval, whereas
  \(\Delta t=4.5~\mathrm{fs}\) gives a rapid numerical instability near
  the start of the trajectory.  This identifies \(4.0~\mathrm{fs}\) as a
  stable practical time step for the present \(0\!-\!200~\mathrm{cm}^{-1}\)
  SO3LR FIMD setup, while \(4.5~\mathrm{fs}\) exceeds the stability
  threshold.
  }
  \label{fig:supp-timestep-stability}
\end{figure}
\section{Mathematical basis for physical fidelity of FIMD}

We consider a classical molecular system with Cartesian coordinates $\mathbf r\in\mathbb R^{3N}$,
momenta $\mathbf p\in\mathbb R^{3N}$, and mass matrix $M=\mathrm{diag}(m_1 I_3,\dots,m_N I_3)$.
The Hamiltonian is
\begin{equation}
H(\mathbf r,\mathbf p)=\tfrac12 \mathbf p^{\mathsf T}M^{-1}\mathbf p + V(\mathbf r),
\end{equation}
and the corresponding Liouville equation on phase-space densities $\rho(\mathbf r,\mathbf p,t)$ reads $\partial_t\rho = \hat{\mathcal L}^\dagger \rho$.

FIMD is formulated around a fixed reference geometry $\mathbf r_0$.
Let $H_0 := \nabla^2 V(\mathbf r_0)$ be the Cartesian Hessian and define the mass-weighted Hessian
$\widetilde H_0 := M^{-1/2}H_0M^{-1/2}$.
Diagonalising,
\begin{equation}
\widetilde H_0 = W\,\Omega^2 W^{\mathsf T},
\qquad
\Omega=\mathrm{diag}(\Omega_1,\dots,\Omega_{n^\nu}),
\end{equation}
yields an orthonormal mode matrix $W$ (restricted to the vibrational subspace, obtained by an existing trajectory) and harmonic frequencies $\{\Omega_\nu\}$.
Normal-mode coordinates $(\mathbf Q,\mathbf P)\in\mathbb R^{n^\nu}\times\mathbb R^{n^\nu}$ are defined by
\begin{equation}
\mathbf Q := W^{\mathsf T}M^{1/2}(\mathbf r-\mathbf r_0),
\qquad
\mathbf P := W^{\mathsf T}M^{-1/2}\mathbf p,
\end{equation}
so that the quadratic (harmonic) reference Hamiltonian is
\begin{equation}
H_0(\mathbf Q,\mathbf P)=\sum_{\nu=1}^{n^\nu}\left(\tfrac12 P_\nu^2+\tfrac12 \Omega_\nu^2 Q_\nu^2\right)+\mathrm{const}.
\end{equation}
The residual potential is
\begin{equation}
\Delta V(\mathbf r)=V(\mathbf r)-V(\mathbf r_0)-\tfrac12(\mathbf r-\mathbf r_0)^{\mathsf T}H_0(\mathbf r-\mathbf r_0),
\qquad
\mathbf F_\Delta(\mathbf r)=-\nabla \Delta V(\mathbf r).
\end{equation}

A band $B$ (a set of mode indices) defines a projector $P_B$ on the normal-mode space and a band-reconstruction
\begin{equation}
\mathbf r_B(\mathbf Q_B)=\mathbf r_0 + M^{-1/2}W_B\,\mathbf Q_B,
\end{equation}
where $W_B$ is the restriction of $W$ to columns in $B$ and $\mathbf Q_B=P_B\mathbf Q$.
Band-limited FIMD evolves $(\mathbf Q_B,\mathbf P_B)$ by composing:
(i) an \emph{exact} harmonic phase advance in the active band generated by $H_0$ restricted to $B$, and
(ii) momentum kicks from the projected residual force $P_B\bigl(W^{\mathsf T}M^{-1/2}\mathbf F_\Delta(\mathbf r_B)\bigr)$.
In NVT, an additional Ornstein--Uhlenbeck (OU) step is applied to $\mathbf P_B$.
The statements below formalise why these building blocks reproduce the expected harmonic geometry and canonical statistics.

\begin{theorem}[Reference harmonic drift torus]\label{thm:drift_torus}
Fix a reference configuration $\mathbf r_0$ with positive-definite vibrational Hessian $H_0$
and eigenpairs $\{(\Omega_\nu,w_\nu)\}_{\nu=1}^{n^\nu}$.
Consider the harmonic reference dynamics (equivalently, the FIMD drift substep with $S_\nu\equiv 0$)
\[
\dot q_\nu = p_\nu,\qquad \dot p_\nu = -\Omega_\nu^2 q_\nu,
\]
whose exact time-$\Delta t$ map in each $(q_\nu,p_\nu)$ plane is the rotation
\[
\binom{q_\nu^+}{p_\nu^+}=
\begin{pmatrix}
\cos(\Omega_\nu\Delta t) & \Omega_\nu^{-1}\sin(\Omega_\nu\Delta t)\\
-\Omega_\nu\sin(\Omega_\nu\Delta t) & \cos(\Omega_\nu\Delta t)
\end{pmatrix}
\binom{q_\nu}{p_\nu}.
\]
Then the quadratic action
\[
I_\nu := p_\nu^2+\Omega_\nu^2 q_\nu^2
\]
is invariant under this drift map for every $\nu$, so each level set $\{I_\nu=\text{const.}\}$
is diffeomorphic to a circle $S^1$ in the $(q_\nu,p_\nu/\Omega_\nu)$ plane. Consequently, for any
choice of constants $\{c_\nu>0\}_{\nu=1}^{n^\nu}$, the invariant set
\[
\mathcal T(\mathbf c):=\bigcap_{\nu=1}^{n^\nu}\{I_\nu=c_\nu\}
\]
is a product of circles,
\[
\mathcal T(\mathbf c)\;\cong\; \underbrace{S^1 \times \cdots \times S^1}_{n^v}
\]
\end{theorem}

\begin{remark}[Full FIMD step perturbs the drift tori]\label{rem:torus_perturbed}
In full FIMD, the drift is composed with source kicks (and optionally thermostatting), so the
invariants $I_\nu$ are not conserved in general: the kicks move the state between the drift tori
$\mathcal T(\mathbf c)$. In the limit $S_\nu\to 0$ (and $\gamma\to 0$ when present), the method
reduces to the exact harmonic rotations above and preserves the invariant circles/tori.
\end{remark}

\begin{proof}
For a single mode $\nu=1$, the reference harmonic drift (i.e.\ the FIMD drift substep with $S_1\equiv 0$)
is the exact solution of
\[
\dot q_1 = p_1,\qquad \dot p_1 = -\Omega_1^2 q_1.
\]
Define the quadratic action
\[
I_1 := p_1^2+\Omega_1^2 q_1^2 .
\]
Differentiating along the flow gives
\[
\frac{dI_1}{dt}
= 2p_1\dot p_1 + 2\Omega_1^2 q_1 \dot q_1
= 2p_1(-\Omega_1^2 q_1) + 2\Omega_1^2 q_1(p_1)
=0,
\]
so $I_1$ is conserved. Hence the level set $\{I_1=c_1\}$ is an ellipse in the $(q_1,p_1)$ plane, which
becomes a circle in the scaled coordinates $(q_1,p_1/\Omega_1)$; therefore $\{I_1=c_1\}\cong S^1$.

Now let $n^\nu\ge 2$ and consider the decoupled reference dynamics
\[
\dot q_\nu = p_\nu,\qquad \dot p_\nu = -\Omega_\nu^2 q_\nu,\qquad \nu=1,\dots,n^\nu.
\]
For each $\nu$, the same calculation yields $\frac{d}{dt}I_\nu=0$ for
$I_\nu:=p_\nu^2+\Omega_\nu^2 q_\nu^2$. Therefore, for any constants $\mathbf c=(c_1,\dots,c_{n^\nu})$
with $c_\nu>0$, the invariant set
\[
\mathcal T(\mathbf c)=\bigcap_{\nu=1}^{n^\nu}\{I_\nu=c_\nu\}
\]
is the Cartesian product of the individual level sets:
\[
\mathcal T(\mathbf c)=\prod_{\nu=1}^{n^\nu}\{I_\nu=c_\nu\}.
\]
Since each factor $\{I_\nu=c_\nu\}\cong S^1$, it follows that
\[
\mathcal T(\mathbf c)\cong \underbrace{S^1 \times \cdots \times S^1}_{n^v},
\]
which is the claimed product-of-circles (torus) structure for the reference harmonic drift.
\end{proof}

\paragraph{Connection to the exact harmonic drift (phase advance).}
For a fixed reference $\mathbf r_0$, the quadratic Hamiltonian $H_0$ is diagonal in normal modes and each retained
mode evolves as an independent harmonic oscillator,
\[
\dot q_\nu = p_\nu,\qquad \dot p_\nu = -\Omega_\nu^2 q_\nu .
\]
Equivalently, in scaled coordinates $(q_\nu,\pi_\nu)$ with $\pi_\nu:=p_\nu/\Omega_\nu$, the drift is a rigid rotation
in the $(q_\nu,\pi_\nu)$ plane by angle $\Omega_\nu\Delta t$, i.e.
\begin{equation}
\binom{q_\nu^+}{p_\nu^+}=
\begin{pmatrix}
\cos(\Omega_\nu\Delta t) & \Omega_\nu^{-1}\sin(\Omega_\nu\Delta t)\\
-\Omega_\nu\sin(\Omega_\nu\Delta t) & \cos(\Omega_\nu\Delta t)
\end{pmatrix}
\binom{q_\nu}{p_\nu}.
\end{equation}
Thus the drift step is \emph{exact} and \emph{symplectic}, and it preserves the quadratic actions
$I_\nu=p_\nu^2+\Omega_\nu^2 q_\nu^2$ and the associated drift tori of Theorem~\ref{thm:drift_torus}.
This exact harmonic phase advance is the reference propagation used inside FIMD; consequently, departures from the
drift-torus geometry arise only from the residual-force kick (and from thermostatting, when present).

\begin{lemma}[Canonical stationarity and kinetic equipartition]\label{lem:canonical_stationarity}
Let $\hat{\mathcal L}$ denote the Liouville operator associated with the exact Hamiltonian flow generated by $H$,
and let $\hat{\mathcal L}^\dagger$ be its adjoint with respect to the Liouville measure $\mathrm d\Gamma$.

\medskip\noindent
\textup{(i)} The canonical density $\rho_\beta = Z^{-1}e^{-\beta H}$ is stationary under the exact Hamiltonian dynamics:
\begin{equation}
\hat{\mathcal L}^\dagger \rho_\beta = 0.
\end{equation}

\medskip\noindent
\textup{(ii)} Assume $H$ is separable with quadratic kinetic energy in the modal momenta,
\begin{equation}
H(q,p)=\sum_{\nu=1}^{n^\nu}\tfrac12\,p_\nu^2 + V(q) + \mathrm{const}.
\end{equation}
Then the momentum marginal factorises as
\begin{equation}
\rho_\beta^{(p)}(\{p_\nu\}) \propto \exp\!\left(-\beta\sum_{\nu=1}^{n^\nu}\tfrac12 p_\nu^2\right),
\end{equation}
so each $p_\nu$ is a one-dimensional Gaussian and
\begin{equation}
\bigl\langle \tfrac12 p_\nu^2 \bigr\rangle_{\rho_\beta} = \tfrac12 k_B T
\qquad \text{for each } \nu .
\end{equation}

\medskip\noindent
\textup{(iii)} If, in addition, $V$ is quadratic (harmonic reference),
$V(q)=\sum_{\nu=1}^{n^\nu}\tfrac12\Omega_\nu^2 q_\nu^2+\mathrm{const}$, then also
$\bigl\langle \tfrac12 \Omega_\nu^2 q_\nu^2 \bigr\rangle_{\rho_\beta} = \tfrac12 k_B T$.
\end{lemma}

\begin{proof}
\emph{(i) Stationarity of the canonical density.}
By construction, the exact Fourier--Liouville operator $\hat{\mathcal L}$ coincides with the classical
Hamiltonian Liouville operator written in the chosen canonical variables. Acting on a smooth observable
$f\colon\mathcal{P}\to\mathbb{R}$, it can be written in Poisson--bracket form
\begin{equation}
  \hat{\mathcal L} f
  = \{f,H\}
  := \sum_{i=1}^{3N}
    \left(
      \frac{\partial f}{\partial r_i}\frac{\partial H}{\partial p_i}
      - \frac{\partial f}{\partial p_i}\frac{\partial H}{\partial r_i}
    \right).
\end{equation}
The adjoint $\hat{\mathcal L}^\dagger$ with respect to the Lebesgue (Liouville) measure
$\mathrm{d}\Gamma = \mathrm{d}\mathbf{r}\,\mathrm{d}\mathbf{p}$ is defined by
\begin{equation}
  \int_{\mathcal{P}} (\hat{\mathcal L} f)\,\rho\,\mathrm{d}\Gamma
  = \int_{\mathcal{P}} f\,(\hat{\mathcal L}^\dagger \rho)\,\mathrm{d}\Gamma
\end{equation}
for all smooth compactly supported $f$.
A standard integration--by--parts calculation, and the antisymmetry of the Poisson bracket give
\begin{equation}
  \hat{\mathcal L}^\dagger \rho
  = -\{\rho,H\}.
\end{equation}
In particular, if $\rho$ is a smooth function of $H$ only, $\rho = \varphi(H)$ for some
$\varphi\in C^1(\mathbb{R})$, then the chain rule yields
\begin{equation}
  \{\varphi(H),H\}
  = \varphi'(H)\,\{H,H\}
  = 0,
\end{equation}
since $\{H,H\}=0$ identically. Choosing $\varphi(x)=Z^{-1}e^{-\beta x}$ gives $\rho_\beta$, hence
\begin{equation}
  \hat{\mathcal L}^\dagger \rho_\beta
  = -\{\rho_\beta,H\}
  = 0,
\end{equation}
which proves stationarity of the canonical density.

\medskip\noindent
\emph{(ii) Active-band momentum marginal and kinetic equipartition.}
One can assume that the Hamiltonian is in the canonical modal variables $(q,p)$ with quadratic kinetic energy,
\begin{equation}
H(q,p)=\sum_{\nu=1}^{n^\nu}\tfrac12 p_\nu^2 + V(q) + \mathrm{const.}
\end{equation}
by the construction of FIMD.
Let $B\subset\{1,\dots,n^\nu\}$ denote an active band and write $B^c$ for its complement.
With $\mathrm d\Gamma=\prod_{\nu=1}^{n^\nu}\mathrm dq_\nu\,\mathrm dp_\nu$, the canonical density is
\begin{equation}
  \rho_\beta(q,p)
  = Z^{-1}\exp\!\left(-\beta\sum_{\nu=1}^{n^\nu}\tfrac12 p_\nu^2\right)\exp\!\left(-\beta V(q)\right).
\end{equation}

Define the active-band momentum marginal by integrating out \emph{all} coordinates and the inactive
momenta $\{p_\mu\}_{\mu\in B^c}$:
\begin{equation}
  \rho_{\beta,B}^{(p)}(p_B)
  := \int_{\mathbb R^{n^\nu}}\int_{\mathbb R^{|B^c|}}
      \rho_\beta(q,p)\,
      \prod_{\nu=1}^{n^\nu}\mathrm dq_\nu\,
      \prod_{\mu\in B^c}\mathrm dp_\mu,
  \qquad p_B:=\{p_\nu\}_{\nu\in B}.
\end{equation}
Using the factorisation
\[
\exp\!\left(-\beta\sum_{\nu=1}^{n^\nu}\tfrac12 p_\nu^2\right)
=
\exp\!\left(-\beta\sum_{\nu\in B}\tfrac12 p_\nu^2\right)\,
\exp\!\left(-\beta\sum_{\mu\in B^c}\tfrac12 p_\mu^2\right),
\]
we obtain
\begin{align}
  \rho_{\beta,B}^{(p)}(p_B)
  &= Z^{-1}\exp\!\left(-\beta\sum_{\nu\in B}\tfrac12 p_\nu^2\right)
     \left[\int_{\mathbb R^{n^\nu}}e^{-\beta V(q)}\prod_{\nu=1}^{n^\nu}\mathrm dq_\nu\right]
     \left[\int_{\mathbb R^{|B^c|}}\exp\!\left(-\beta\sum_{\mu\in B^c}\tfrac12 p_\mu^2\right)
           \prod_{\mu\in B^c}\mathrm dp_\mu\right].
\end{align}
Both bracketed integrals are finite constants independent of $p_B$ (and can be absorbed into the
normalisation). Hence
\begin{equation}
  \rho_{\beta,B}^{(p)}(p_B)
  \propto \exp\!\left(-\beta\sum_{\nu\in B}\tfrac12 p_\nu^2\right)
  = \prod_{\nu\in B}\exp\!\left(-\beta\tfrac12 p_\nu^2\right),
\end{equation}
so the active-band momenta are independent one-dimensional Gaussians:
\begin{equation}
  \rho_{\beta,B}^{(p)}(p_B)=\prod_{\nu\in B}\rho_{\beta}(p_\nu),\qquad
  \rho_{\beta}(p)=C\,\exp\!\left(-\beta\tfrac12 p^2\right),
\end{equation}
with the same constant $C$ for all $\nu\in B$.

Finally, for any $\nu\in B$,
\begin{equation}
  \left\langle \tfrac12 p_\nu^2 \right\rangle_{\rho_\beta}
  = \int_{\mathbb R}\tfrac12 p^2\,\rho_{\beta}(p)\,\mathrm dp
  = \frac{1}{2\beta}
  = \tfrac12 k_BT,
\end{equation}
which is the mode-by-mode kinetic equipartition on the active band.

\medskip\noindent
\emph{(iii) Coordinate equipartition in the harmonic reference.}
Assume the harmonic reference Hamiltonian
\[
H_0(q,p)=\sum_{\nu=1}^{n^\nu}\Big(\tfrac12 p_\nu^2+\tfrac12\Omega_\nu^2 q_\nu^2\Big)+\mathrm{const.}
\]
and consider a fixed mode $\nu$.  Introduce the drift invariant (Theorem~\ref{thm:drift_torus})
\[
I_\nu:=p_\nu^2+\Omega_\nu^2 q_\nu^2.
\]
By Theorem~\ref{thm:drift_torus}, each level set $\{I_\nu=c\}$ is a circle, and the full invariant set
$\bigcap_{\nu=1}^{n^\nu}\{I_\nu=c_\nu\}$ is a product of circles.  This motivates polar (angle) coordinates on
each modal plane.  Define the scaled momentum $\pi_\nu:=p_\nu/\Omega_\nu$, so that
$I_\nu=\Omega_\nu^2(q_\nu^2+\pi_\nu^2)$, and write
\[
q_\nu=\rho_\nu\cos\theta_\nu,\qquad \pi_\nu=\rho_\nu\sin\theta_\nu,
\qquad \rho_\nu\ge 0,\ \theta_\nu\in[0,2\pi).
\]
Equivalently, $I_\nu=\Omega_\nu^2\rho_\nu^2$, and $(\rho_\nu,\theta_\nu)$ parameterises the circles from
Theorem~\ref{thm:drift_torus}.  The Liouville measure in this mode becomes
\[
\mathrm dq_\nu\,\mathrm dp_\nu
= \Omega_\nu\,\mathrm dq_\nu\,\mathrm d\pi_\nu
= \Omega_\nu\,(\rho_\nu\,\mathrm d\rho_\nu\,\mathrm d\theta_\nu)
= \frac{1}{2\Omega_\nu}\,\mathrm dI_\nu\,\mathrm d\theta_\nu,
\]
since $\rho_\nu\,\mathrm d\rho_\nu=\mathrm dI_\nu/(2\Omega_\nu^2)$.

Because $H_0$ is a sum of modes and, for this mode,
$\tfrac12(p_\nu^2+\Omega_\nu^2 q_\nu^2)=\tfrac12 I_\nu$, the canonical density factor for $(I_\nu,\theta_\nu)$ is
\[
\propto \exp\!\big(-\beta\tfrac12 I_\nu\big),
\]
which is independent of $\theta_\nu$.  Hence $\theta_\nu$ is uniform on $[0,2\pi)$ under $\rho_\beta$, and
\[
\Omega_\nu^2 q_\nu^2 = \Omega_\nu^2\rho_\nu^2\cos^2\theta_\nu = I_\nu\cos^2\theta_\nu.
\]
Taking the canonical average and integrating out the uniform angle gives
\[
\big\langle \Omega_\nu^2 q_\nu^2 \big\rangle_{\rho_\beta}
= \big\langle I_\nu \big\rangle_{\rho_\beta}\,
\big\langle \cos^2\theta_\nu \big\rangle
= \big\langle I_\nu \big\rangle_{\rho_\beta}\cdot \frac12.
\]
It remains to compute $\langle I_\nu\rangle$.  Since the $(I_\nu,\theta_\nu)$ density is proportional to
$e^{-\beta I_\nu/2}$ and the Jacobian contributes only the constant factor $(2\Omega_\nu)^{-1}$,
the marginal of $I_\nu$ is an exponential distribution on $[0,\infty)$ with mean
\[
\big\langle I_\nu \big\rangle_{\rho_\beta}=\frac{2}{\beta}.
\]
Therefore,
\[
\left\langle \tfrac12\Omega_\nu^2 q_\nu^2\right\rangle_{\rho_\beta}
= \tfrac12\cdot \tfrac12\,\big\langle I_\nu \big\rangle_{\rho_\beta}
= \frac{1}{2\beta}
= \tfrac12 k_BT,
\]
which is the coordinate equipartition statement mode-by-mode.

\end{proof}

Lemma~\ref{lem:canonical_stationarity} states two foundational facts:
\begin{itemize}
    \item[(i)] for \emph{exact} Hamiltonian flow, the canonical density $\rho_\beta\propto e^{-\beta H}$ is stationary; and
    \item[(ii)] for any separable Hamiltonian with quadratic kinetic energy, the momentum marginal (and, in particular, any active-band marginal) factorises into independent Gaussians and obeys mode-by-mode kinetic equipartition.
\end{itemize}
These two properties are the statistical targets that an NVT integrator must preserve.

FIMD intentionally evolves only an active subset of modes.
Given a band $B$, we split normal-mode variables into active and inactive components,
$(\mathbf Q,\mathbf P)=(\mathbf Q_B,\mathbf Q_{B^c},\mathbf P_B,\mathbf P_{B^c})$.
The natural canonical object on the active subspace is the band marginal
\begin{equation}
\rho_{\beta,B}(\mathbf Q_B,\mathbf P_B)
:=
\int \rho_\beta(\mathbf Q_B,\mathbf Q_{B^c},\mathbf P_B,\mathbf P_{B^c})\,
\mathrm d\mathbf Q_{B^c}\,\mathrm d\mathbf P_{B^c}.
\end{equation}
For any separable Hamiltonian with quadratic kinetic energy, the active-band momentum marginal
factorises into independent Gaussians for $\nu\in B$.
Accordingly, the correct band temperature diagnostic is
\begin{equation}
T_B(t) := \frac{2}{|B|k_B}\,\mathrm{KE}_B(t)
\quad\text{with}\quad
\mathrm{KE}_B(t)=\sum_{\nu\in B}\tfrac12 P_\nu(t)^2,
\end{equation}
and the expected equilibrium value is $\langle T_B\rangle = T$ by kinetic equipartition.
Theorem~\ref{thm:band_stationarity} makes these statements precise at the level of the band Liouville operator.

\begin{theorem}[Band-canonical stationarity and band kinetic equipartition]\label{thm:band_stationarity}
Let $B\subset\{1,\dots,n^\nu\}$ be an active band and let $\rho_{\beta,B}$ denote the canonical band marginal
\[
\rho_{\beta,B}(q_B,p_B):=\int \rho_\beta(q_B,q_{B^c},p_B,p_{B^c})\,\mathrm dq_{B^c}\,\mathrm dp_{B^c}.
\]

\medskip\noindent
\textup{(1) (Reference drift) } Let $\hat{\mathcal L}_{0,B}$ be the Liouville operator for the \emph{reference harmonic drift}
restricted to $B$, i.e.\ the Hamiltonian flow generated by
\[
H_{0,B}(q_B,p_B)=\sum_{\nu\in B}\Big(\tfrac12 p_\nu^2+\tfrac12\Omega_\nu^2 q_\nu^2\Big).
\]
Then the band-canonical density $\rho_{\beta,B}^{(0)}\propto e^{-\beta H_{0,B}}$ is stationary:
\[
\hat{\mathcal L}_{0,B}^\dagger \rho_{\beta,B}^{(0)}=0.
\]

\medskip\noindent
\textup{(2) (Kinetic equipartition on the band) } For any Hamiltonian of the form
$H(q,p)=\sum_{\nu=1}^{n^\nu}\tfrac12 p_\nu^2+V(q)+\mathrm{const}$,
the band momentum marginal factorises into independent Gaussians, and
\[
\Bigl\langle \mathrm{KE}_B \Bigr\rangle_{\rho_{\beta,B}}
=\Bigl\langle \sum_{\nu\in B}\tfrac12 p_\nu^2 \Bigr\rangle_{\rho_{\beta,B}}
=\frac{|B|}{2}\,k_B T,
\qquad
\Bigl\langle \tfrac12 p_\nu^2 \Bigr\rangle_{\rho_{\beta,B}}=\tfrac12 k_B T\ \ (\nu\in B).
\]
\end{theorem}

\begin{proof}
\emph{(1) Stationarity.}
Let $H_{0}$ denote the harmonic reference Hamiltonian written in normal-mode coordinates,
\[
H_{0}(q,p)=\sum_{\nu=1}^{n^\nu}\Big(\tfrac12 p_\nu^2+\tfrac12\Omega_\nu^2 q_\nu^2\Big)+\mathrm{const.},
\]
and let $\hat{\mathcal L}_{0}$ be its Liouville operator. By Lemma~\ref{lem:canonical_stationarity}(i),
the corresponding canonical density $\rho_{\beta}^{(0)}:=Z_0^{-1}e^{-\beta H_{0}}$ satisfies
\begin{equation}
  \hat{\mathcal L}_{0}^\dagger \rho_{\beta}^{(0)} = 0.
\end{equation}
Because $H_{0}$ is a sum of independent modal Hamiltonians, $\hat{\mathcal L}_{0}$ decomposes as a sum of
commuting modal generators and hence splits over the band:
\begin{equation}
  \hat{\mathcal L}_{0}
  =
  \sum_{\nu=1}^{n^\nu}\hat{\mathcal L}_{0,\nu}
  =
  \hat{\mathcal L}_{0,B}+\hat{\mathcal L}_{0,B^c},
\end{equation}
where $\hat{\mathcal L}_{0,B}:=\sum_{\nu\in B}\hat{\mathcal L}_{0,\nu}$ and similarly for $B^c$.

Let $\varphi\in C_c^\infty(\mathcal P_B)$ be a smooth compactly supported test function on the band phase
space $\mathcal P_B$, and define its trivial extension
\[
\widetilde{\varphi}(q,p):=\varphi(q_B,p_B),
\]
independent of $(q_{B^c},p_{B^c})$. Then $\hat{\mathcal L}_{0,B^c}\widetilde{\varphi}=0$, and therefore
\begin{align}
0
&= \int_{\mathcal P} (\hat{\mathcal L}_{0}\widetilde{\varphi})\,\rho_{\beta}^{(0)}\,\mathrm d\Gamma
 = \int_{\mathcal P} (\hat{\mathcal L}_{0,B}\widetilde{\varphi})\,\rho_{\beta}^{(0)}\,\mathrm d\Gamma \\
&= \int_{\mathcal P_B} (\hat{\mathcal L}_{0,B}\varphi)(q_B,p_B)
   \left[\int_{\mathcal P_{B^c}}\rho_{\beta}^{(0)}(q,p)\,\mathrm d\Gamma_{B^c}\right]
   \mathrm d\Gamma_B.
\end{align}
Define the band canonical density for the reference system as the marginal
\[
\rho_{\beta,B}^{(0)}(q_B,p_B):=\int_{\mathcal P_{B^c}}\rho_{\beta}^{(0)}(q,p)\,\mathrm d\Gamma_{B^c}.
\]
Then the previous identity becomes
\[
\int_{\mathcal P_B} (\hat{\mathcal L}_{0,B}\varphi)\,\rho_{\beta,B}^{(0)}\,\mathrm d\Gamma_B=0
\qquad \text{for all } \varphi\in C_c^\infty(\mathcal P_B).
\]
By the definition of the adjoint $\hat{\mathcal L}_{0,B}^\dagger$ on $\mathcal P_B$, this implies
\[
\hat{\mathcal L}_{0,B}^\dagger \rho_{\beta,B}^{(0)}=0,
\]
i.e.\ the band-canonical density is stationary under the band-restricted (reference harmonic) dynamics.

\medskip\noindent
\emph{(2) Equipartition in the band.}
Assume the full Hamiltonian has quadratic kinetic energy in the modal momenta,
\[
H(q,p)=\sum_{\nu=1}^{n^\nu}\tfrac12 p_\nu^2 + V(q)+\mathrm{const.}
\]
Then, exactly as in Lemma~\ref{lem:canonical_stationarity}(ii), the momentum part of the canonical density
factorises as $\exp(-\beta\sum_\nu \tfrac12 p_\nu^2)$, irrespective of $V(q)$.
Define the active-band momentum marginal by integrating out all coordinates and inactive momenta:
\[
\rho_{\beta,B}^{(p)}(p_B)
:=
\int_{\mathbb R^{n^\nu}}\int_{\mathbb R^{|B^c|}}
\rho_\beta(q,p)\,
\prod_{\nu=1}^{n^\nu}\mathrm dq_\nu\,
\prod_{\mu\in B^c}\mathrm dp_\mu .
\]
Using the factorisation
\[
\exp\!\left(-\beta\sum_{\nu=1}^{n^\nu}\tfrac12 p_\nu^2\right)
=
\exp\!\left(-\beta\sum_{\nu\in B}\tfrac12 p_\nu^2\right)\,
\exp\!\left(-\beta\sum_{\mu\in B^c}\tfrac12 p_\mu^2\right),
\]
the integrations over $q$ and $p_{B^c}$ contribute only constants independent of $p_B$, hence
\[
\rho_{\beta,B}^{(p)}(p_B)\propto
\exp\!\left(-\beta\sum_{\nu\in B}\tfrac12 p_\nu^2\right)
=
\prod_{\nu\in B}\exp\!\left(-\beta\tfrac12 p_\nu^2\right).
\]
Therefore $\{p_\nu\}_{\nu\in B}$ are independent one-dimensional Gaussians, and for each $\nu\in B$,
\[
\Bigl\langle \tfrac12 p_\nu^2 \Bigr\rangle_{\rho_{\beta,B}}
=
\frac{1}{2\beta}
=\tfrac12 k_BT.
\]
Summing over $\nu\in B$ yields
\[
\Bigl\langle \mathrm{KE}_B \Bigr\rangle_{\rho_{\beta,B}}
=
\Bigl\langle \sum_{\nu\in B}\tfrac12 p_\nu^2 \Bigr\rangle_{\rho_{\beta,B}}
=
\frac{|B|}{2}\,k_BT,
\]
which is the claimed band equipartition.
\end{proof}

\paragraph{Why a band thermostat is enough (and what it guarantees).}
Band-limited FIMD evolves only the active normal-mode variables $(\mathbf Q_B,\mathbf P_B)$, while the inactive
coordinates are not propagated explicitly. In an NVT setting, the thermostat therefore cannot (and need not) enforce
the canonical law on the full $(\mathbf Q,\mathbf P)$; instead, the appropriate statistical target is the \emph{band
marginal} on the variables we actually propagate,
\[
\rho_{\beta,B}(\mathbf Q_B,\mathbf P_B)
:=
\int \rho_\beta(\mathbf Q_B,\mathbf Q_{B^c},\mathbf P_B,\mathbf P_{B^c})\,
\mathrm d\mathbf Q_{B^c}\,\mathrm d\mathbf P_{B^c}.
\]
The key point is that the canonical distribution imposes a universal constraint on momenta that does \emph{not}
depend on the detailed form of the potential: for any Hamiltonian with quadratic kinetic energy in the canonical
momenta,
\[
H(\mathbf Q,\mathbf P)=\sum_{\nu}\tfrac12 P_\nu^2 + V(\mathbf Q) + \mathrm{const.},
\]
the momentum factor in $\rho_\beta\propto e^{-\beta H}$ is
$\exp(-\beta\sum_\nu \tfrac12 P_\nu^2)$, and hence the \emph{band momentum marginal} is
\[
\rho^{(P)}_{\beta,B}(\mathbf P_B)\ \propto\ \exp\!\left(-\beta\sum_{\nu\in B}\tfrac12 P_\nu^2\right),
\]
i.e.\ independent one-dimensional Gaussians for each $\nu\in B$ with
$\langle \tfrac12 P_\nu^2\rangle = \tfrac12 k_B T$ and
$\langle \mathrm{KE}_B\rangle = \sum_{\nu\in B}\langle \tfrac12 P_\nu^2\rangle = \tfrac{|B|}{2}k_B T$.
This is the precise sense in which a thermostat acting only on $\mathbf P_B$ is sufficient for controlling the
band temperature.

Theorem~\ref{thm:band_stationarity} plays a complementary role for the \emph{reference drift} used inside FIMD.
For the harmonic reference Hamiltonian on the band,
\[
H_{0,B}(\mathbf Q_B,\mathbf P_B)=\sum_{\nu\in B}\Big(\tfrac12 P_\nu^2+\tfrac12\Omega_\nu^2 Q_\nu^2\Big),
\]
the band-canonical density $\propto e^{-\beta H_{0,B}}$ is stationary under the band-restricted Hamiltonian flow, and the drift map is an exact symplectic phase rotation (Theorem~\ref{thm:drift_torus}). 
Thus the drift step introduces no sampling bias by itself. In a general anharmonic system, deviations from the reference drift arise only through the residual-force kick, which couples the active subspace to eliminated coordinates implicitly; the band thermostat is then used to (i) prevent systematic heating/cooling of the propagated modes and (ii) restore the correct Gaussian statistics for $\mathbf P_B$, while leaving the exact drift rotation untouched.

Operationally, this yields two concrete NVT diagnostics that we can test directly:
\begin{enumerate}
\item \emph{Band kinetic equipartition:}
\[
\left\langle \tfrac12 P_\nu^2\right\rangle = \tfrac12 k_B T\quad(\nu\in B),
\qquad
\left\langle \mathrm{KE}_B\right\rangle = \frac{|B|}{2}k_B T,
\]
equivalently $\langle T_B\rangle = T$ for $T_B(t):=\frac{2}{|B|k_B}\mathrm{KE}_B(t)$.

\item \emph{Gaussian (Maxwell--Boltzmann) statistics of reconstructed velocities:} the band-limited Cartesian velocity field is a linear reconstruction from the band momenta,
\[
\dot{\mathbf r}_B(t)
=\sum_{\nu\in B}\dot Q_\nu(t)\, \mathbf w_\nu
=\sum_{\nu\in B} P_\nu(t)\,\mathbf u_\nu,
\]
for fixed reconstruction vectors $\mathbf u_\nu$ (mass-weight factors absorbed into $\mathbf u_\nu$).
Since linear combinations of independent Gaussians are Gaussian, the canonical law for $\mathbf P_B$ implies that each Cartesian component of $\dot{\mathbf r}_B$ is Gaussian. 
In practice we therefore test that the normalised velocity components have an approximately standard normal distribution.
\end{enumerate}
We emphasise that these guarantees concern the \emph{band momentum law} (and thus the band kinetic temperature).
In a truncated, strongly anharmonic system this does not, by itself, imply exact canonical sampling of the band \emph{coordinates} $\mathbf Q_B$; rather, it provides the correct and testable NVT control on the degrees of freedom that are explicitly propagated by band-limited FIMD.

\begin{theorem}[Band Langevin invariance and Gaussian velocity law]\label{thm:band_langevin_MB}
Let $B$ be an active band. Consider a band Hamiltonian on $(q_B,p_B)$ of the separable form
\begin{equation}
H_B(q_B,p_B)=\sum_{\nu\in B}\tfrac12 p_\nu^2 + U_B(q_B),
\end{equation}
with associated Hamiltonian Liouville operator $\hat{\mathcal L}_{H_B}$.
Let $\hat{\mathcal L}_{\gamma,B}$ be the Ornstein--Uhlenbeck generator acting only on the band momenta,
\begin{equation}
\hat{\mathcal L}_{\gamma,B}
=
\sum_{\nu\in B}
\left(
-\gamma p_\nu\,\partial_{p_\nu}
+\gamma k_B T\,\partial_{p_\nu}^2
\right).
\end{equation}
Then the band canonical density
\begin{equation}
\rho_{\beta,B}^{(B)}(q_B,p_B) = Z_B^{-1}\exp\!\bigl(-\beta H_B(q_B,p_B)\bigr)
\end{equation}
is invariant under the band Langevin generator:
\begin{equation}
(\hat{\mathcal L}_{H_B}+\hat{\mathcal L}_{\gamma,B})^\dagger \rho_{\beta,B}^{(B)}=0.
\end{equation}
In particular, the band momenta remain independent Gaussians with variance $k_BT$, and
$\langle \tfrac12 p_\nu^2\rangle=\tfrac12 k_BT$ for all $\nu\in B$.

Moreover, any linear reconstruction of velocities from $p_B$ (e.g.\ the band-limited Cartesian velocity
$\dot{\mathbf r}_B=\sum_{\nu\in B} p_\nu\,\mathbf u_\nu$) is multivariate Gaussian. If $B$ spans the full
vibrational subspace used in the reconstruction (no projection), then each Cartesian component satisfies
$v_i\sim\mathcal N(0,k_BT/m_i)$ and $\langle\tfrac12 m_i v_i^2\rangle=\tfrac12 k_BT$ (Maxwell--Boltzmann).

Finally, the exact discrete OU update
\begin{equation}
p_\nu \leftarrow e^{-\gamma\Delta t}p_\nu
+ \sqrt{k_B T\bigl(1-e^{-2\gamma\Delta t}\bigr)}\,\xi_\nu,
\qquad
\xi_\nu\sim\mathcal N(0,1),
\end{equation}
preserves the same Gaussian momentum marginal at the discrete level.
\end{theorem}

\begin{proof}
We work on the band phase space $\mathcal P_B$ with canonical variables $(q_B,p_B)$.
Let the band Hamiltonian be of separable form
\[
H_B(q_B,p_B)=\sum_{\nu\in B}\tfrac12 p_\nu^2 + U_B(q_B),
\]
with Hamiltonian Liouville operator $\hat{\mathcal L}_{H_B}f=\{f,H_B\}$, and let
$\hat{\mathcal L}_{\gamma,B}$ be the Ornstein--Uhlenbeck generator acting only on $\{p_\nu\}_{\nu\in B}$ as in the theorem.
Denote the band canonical density by
\[
\rho_{\beta,B}^{(B)}(q_B,p_B)=Z_B^{-1}\exp\!\bigl(-\beta H_B(q_B,p_B)\bigr).
\]

\medskip\noindent
\emph{(1) Invariance of $\rho_{\beta,B}^{(B)}$.}
By Lemma~\ref{lem:canonical_stationarity}(i) applied on $\mathcal P_B$, any density of the form $\varphi(H_B)$
is stationary for the Hamiltonian flow generated by $H_B$, hence
\begin{equation}
  \hat{\mathcal L}_{H_B}^\dagger \rho_{\beta,B}^{(B)}=0.
\end{equation}
It remains to show that $\rho_{\beta,B}^{(B)}$ is stationary for the OU part.
Since $U_B$ is independent of $p_B$, the band canonical density factorises as
\begin{equation}
  \rho_{\beta,B}^{(B)}(q_B,p_B)
  = \rho_{\beta,B}^{(q)}(q_B)\,\rho_{\beta,B}^{(p)}(p_B),
  \qquad
  \rho_{\beta,B}^{(p)}(p_B)\propto \exp\!\left(-\beta\sum_{\nu\in B}\tfrac12 p_\nu^2\right).
\end{equation}
Because $\hat{\mathcal L}_{\gamma,B}$ differentiates only with respect to $p_\nu$, we have
\begin{equation}
  \hat{\mathcal L}_{\gamma,B}^\dagger \rho_{\beta,B}^{(B)}
  = \rho_{\beta,B}^{(q)}(q_B)\,\hat{\mathcal L}_{\gamma,B}^\dagger \rho_{\beta,B}^{(p)}(p_B).
\end{equation}
Thus it suffices to prove $\hat{\mathcal L}_{\gamma,B}^\dagger \rho_{\beta,B}^{(p)}=0$.

Fix $\nu\in B$ and write $\rho_{\beta,\nu}(p)\propto e^{-\beta p^2/2}$. The one-dimensional adjoint OU operator is
\begin{equation}
  \hat{\mathcal L}_{\gamma,\nu}^\dagger \rho
  =
  \gamma\,\partial_{p}(p\rho)
  + \gamma k_B T\,\partial_{p}^2\rho.
\end{equation}
A direct computation gives
\[
\partial_{p}\rho_{\beta,\nu}= -\beta p\,\rho_{\beta,\nu},
\qquad
\partial_{p}^2\rho_{\beta,\nu}=(-\beta+\beta^2 p^2)\rho_{\beta,\nu},
\]
and hence
\begin{align}
  \hat{\mathcal L}_{\gamma,\nu}^\dagger \rho_{\beta,\nu}
  &= \gamma\bigl(\rho_{\beta,\nu}+p\,\partial_{p}\rho_{\beta,\nu}\bigr)
     +\gamma k_B T\,\partial_{p}^2\rho_{\beta,\nu} \\
  &= \gamma\rho_{\beta,\nu}\Bigl[(1-\beta k_B T)+(-\beta+\beta^2 k_B T)p^2\Bigr].
\end{align}
With $\beta=(k_B T)^{-1}$ both coefficients vanish, so
\begin{equation}
  \hat{\mathcal L}_{\gamma,\nu}^\dagger \rho_{\beta,\nu}=0.
\end{equation}
Since $\rho_{\beta,B}^{(p)}(p_B)=\prod_{\nu\in B}\rho_{\beta,\nu}(p_\nu)$ and
$\hat{\mathcal L}_{\gamma,B}^\dagger=\sum_{\nu\in B}\hat{\mathcal L}_{\gamma,\nu}^\dagger$, it follows that
\begin{equation}
  \hat{\mathcal L}_{\gamma,B}^\dagger \rho_{\beta,B}^{(p)}=0,
  \qquad\text{and hence}\qquad
  \hat{\mathcal L}_{\gamma,B}^\dagger \rho_{\beta,B}^{(B)}=0.
\end{equation}
By linearity of the adjoint, we conclude
\begin{equation}
(\hat{\mathcal L}_{H_B}+\hat{\mathcal L}_{\gamma,B})^\dagger \rho_{\beta,B}^{(B)}=0.
\end{equation}

\medskip\noindent
\emph{(2) Modal Gaussians and induced velocity statistics.}
From the factorised momentum density \begin{equation*}
    \rho_{\beta,B}^{(p)}(p_B)\propto
\exp(-\beta\sum_{\nu\in B}\tfrac12 p_\nu^2)
\end{equation*} we immediately obtain, for each $\nu\in B$,
\begin{equation}
  p_\nu \sim \mathcal N(0,k_B T),
  \qquad
  \Bigl\langle \tfrac12 p_\nu^2 \Bigr\rangle_{\rho_{\beta,B}^{(B)}}=\tfrac12 k_B T,
\end{equation}
and summing over $\nu\in B$ yields $\langle \mathrm{KE}_B\rangle = \frac{|B|}{2}k_BT$.

Let $W_B$ denote the mass-weighted mode matrix restricted to the band $B$, and set
$\mathbf u:=M^{1/2}\mathbf v$ so that, for the (band) reconstruction used here,
\begin{equation}
  \mathbf u = W_B\,p_B.
\end{equation}
Since $p_B$ is centred Gaussian with covariance $k_B T\,I$, $\mathbf u$ is centred Gaussian with covariance
\begin{equation}
  \Sigma_u = k_B T\,W_B W_B^\top.
\end{equation}
Transforming back to Cartesian velocities $\mathbf v=M^{-1/2}\mathbf u$ gives a centred Gaussian
\begin{equation}
  \mathbf v \sim \mathcal N(0,\Sigma_v),
  \qquad
  \Sigma_v = k_B T\,M^{-1/2}W_B W_B^\top M^{-1/2}.
\end{equation}
If $B$ spans the complete mode basis used in the reconstruction, so that $W_BW_B^\top=I$ on that space,
then $\Sigma_v=k_B T\,M^{-1}$ and the induced velocity law is Maxwell--Boltzmann:
each Cartesian component satisfies $v_i\sim\mathcal N(0,k_B T/m_i)$ and
$\langle \tfrac12 m_i v_i^2\rangle=\tfrac12 k_B T$.

\medskip\noindent
Finally, the discrete update
\[
p_\nu \leftarrow e^{-\gamma\Delta t}p_\nu
+ \sqrt{k_B T\bigl(1-e^{-2\gamma\Delta t}\bigr)}\,\xi_\nu,
\qquad \xi_\nu\sim\mathcal N(0,1),
\]
is the exact time-$\Delta t$ solution map of the scalar Ornstein--Uhlenbeck process generated by
$\hat{\mathcal L}_{\gamma,\nu}$, and therefore preserves the invariant Gaussian marginal
$p_\nu\sim\mathcal N(0,k_B T)$ at the discrete level.
\end{proof}


\end{document}